\documentclass[11pt,a4paper]{article}

\usepackage{fullpage}
\usepackage{amssymb,amsmath,stmaryrd,amsthm}
\usepackage[round]{natbib}
\usepackage{accents}

\usepackage{enumerate}
\usepackage[mathscr]{eucal}
\usepackage{epsfig}
\usepackage{pstricks,pst-node,pst-text,pst-3d}
\usepackage{pifont}

\usepackage{textcomp}   
\usepackage{multicol}

\theoremstyle{definition}
\newtheorem{theorem}{Theorem}
\newtheorem{lemma}{Lemma}
\newtheorem{remark}{Remark}
\newtheorem{observation}{Observation}

\newtheorem{example}{Example}
\newtheorem{definition}{Definition}
\newtheorem{proposition}{Proposition}

\newtheorem{corollary}{Corollary}

\def\1{\mathbf{1}}
\def\cA {\mathcal{A}}
\def\cB {\mathcal{B}}

\def\cI{\mathcal{I}}

\def\cT{\mathcal{T}}

\def\ZZ{{\mathbb Z}}
\def\NN{{\mathbb N}}

\def \0{{\mathbf{0}}}

\def\scX{\mathscr{X}}

\def\intX{\accentset{\circ}{X}}
\def\cloX{\overline{X}}

\def\int{\mathsf{int}}
\def\clo{\mathsf{clo}}

\def \mymid {\,:\,}

\def\F{F}

\usepackage{color}
\definecolor{rouge}{RGB}{255,0,0}
\definecolor{bleu}{RGB}{0,0,255}
\definecolor{vert}{RGB}{0,123,0}
\definecolor{violet}{RGB}{148,0,211}
\definecolor{cyan}{RGB}{0,206,209}
\definecolor{rose}{RGB}{235,66,109}

\usepackage[textwidth=30mm]{todonotes}

\begin{document}

\title{ Diffusion in large networks
}

\author{Michel GRABISCH\thanks{Corresponding author. Tel
    (+33)14407-8744.} ${}^{(a)(b)}$, Agnieszka
RUSINOWSKA${}^{(a)(c)}$, Xavier VENEL${}^{(d)}$ \\
$^{(a)}$ Centre d'Economie de la Sorbonne, 106-112 Bd de l'H\^{o}pital \\ 75647 Paris, France \\
$^{(b)}$ Paris School of Economics, Universit\'{e} Paris I Panth\'{e}on-Sorbonne \\ 
$^{(c)}$ Paris School of Economics -- CNRS \\
$^{(d)}$ Dipartimento di Economia e Finanza, LUISS \\ Viale Romania 32, 00197 Rome, Italy \\
\normalsize\tt{\{michel.grabisch,agnieszka.rusinowska\}@univ-paris1.fr, xvenel@luiss.it}}

\date{Version of \today}
\maketitle

\begin{abstract}
We investigate the phenomenon of diffusion in a countably infinite society of individuals interacting with their neighbors in a network. At a given time, each individual is either active or inactive. The diffusion is driven by two characteristics: the network structure and the diffusion mechanism represented by an aggregation function. We distinguish between two diffusion mechanisms (probabilistic, deterministic) and focus on two types of aggregation functions (strict, Boolean). Under strict aggregation functions, polarization of the society cannot happen, and its state evolves towards a mixture of infinitely many active and infinitely many inactive agents, or towards a homogeneous society. Under Boolean aggregation functions, the diffusion process becomes deterministic and the contagion model of \cite{mor00} becomes a particular case of our framework. Polarization can then happen. Our dynamics also allows for cycles in both cases. The network structure is not relevant for these questions, but is important for establishing irreducibility, at the price of a richness assumption: the network should contain infinitely many complex stars and have enough space for storing local configurations.
Our model can be given a game-theoretic interpretation via a local coordination game, where each player would apply a best-response strategy in a random neighborhood. 
\end{abstract}

\noindent {\bf JEL Classification}: C7, D7, D85

\vspace{2 mm}

\noindent {\bf Keywords}: 
diffusion, countable network, aggregation function, polarization, convergence, best-response

\newpage

\section{Introduction}

Although the first works on diffusion do not refer to explicit network structure, certain early and widely used models (e.g., \cite{bas69}) do incorporate some aspects of imitation. 
Since then, studying diffusion in networks has a very strong established position in the literature, with a variety of approaches and different ways of modeling diffusion. Adoption of a new technology, product purchasing and marketing, opinion formation, influence,
social learning, information cascades and transmission,
rumors,
fashions, contagion and disease infection, financial contagion,
among many other related phenomena, are all relevant to network diffusion; for surveys on these issues see, e.g., two chapters in \cite{jac08a} on diffusion through networks and learning, \cite{ace-ozd11} for an overview of Bayesian learning and imitation models, several chapters in \cite{bra-gal-rog16} on diffusion, learning and contagion (\cite{lam16}, \cite{gol-sad16}, \cite{cab-gal-got15}, \cite{ace-ozd-tah16}), and \cite{gra-rus20} for a survey on nonstrategic models of opinion dynamics; see also e.g., \cite{rog03} for overviews of the literature on diffusion and its applications, and \cite{sil01} for ``word-of-mouth marketing''. 

As the spread of ideas, products or rumors can be similar to that of infectious disease, many frameworks of diffusion in networks are indeed inspired by the standard epidemiological models such as SI, SIS, and SIR (\cite{bai75}, \cite{kee-roh08}). These models differ, in particular, by the possibility of re-infection after recovering. Contrary to the SIS model, this option is not possible in the SIR model, while in the SI framework once individuals become infected, they remain infected forever. Different variants of the models are proposed, e.g., in the literature on network diffusion in economics (\cite{jac-rog07a}, \cite{lop07}).

\medskip

The present paper investigates diffusion in large networks and is therefore fairly related to the different branches of the network literature mentioned above. Our point of departure is the work of \cite{mor00} on contagion, 
defined in his paper as a diffusion phenomenon occurring when one of two actions can spread from a
finite set of individuals to the whole (countably infinite) population.
The basic mechanism in \cite{mor00} is a deterministic process, known under the name of {\it threshold
  model} (see, e.g., \cite{gra78}, \cite{sch78}), where an agent becomes infected (or {\it
  active}, as we will use this term throughout the paper) if the number of
infected agents in its neighborhood exceeds a given threshold. 
To make the mechanism more realistic, we introduce some randomness and 
replace the threshold mechanism by the following: an agent becomes active with some
probability, which is a function of the status (active or inactive) of every
agent in its neighborhood. The probability is equal to 1 (respectively, 0) if
all agents in the neighborhood are active (respectively, inactive), and it is
required that the probability monotonically increases with the number of active
agents. Such a function determining the probability to be active is called an
{\it aggregation function}. Hence, the new model is in the vein of \cite{gra-gri98}, who study the growth of a cellular probabilistic automaton on $\ZZ^2$, although our perspective is different and we work on arbitrary countable networks, of which $\ZZ^2$ is only a particular case. The approach based on aggregation functions also relates the present paper to previous works of two of the authors that study in depth models of opinion dynamics based on aggregation functions (see \cite{gra-rus13,foe-gra-rus13,gra-poi-rus19}).

Similarly to \cite{mor00,gra-rus13,foe-gra-rus13}, we assume that an agent is not a neighbor of itself. This is a reasonable assumption in an imitation model, for example when studying the diffusion of a technology. An agent tends to adopt or imitate the status/action of its neighbors. In the threshold model, \cite{mor00} suggests a game-theoretic interpretation via a local coordination game, where each player would apply a best-response strategy in a random neighborhood. Since a player does not play against itself, it is then again natural in this interpretation to assume that a player is not a neighbor of itself.

There are two essential differences with our previous works:   
1) First, we assume that the society of agents is not anymore finite but countably
  infinite, like in \cite{mor00}. Indeed, the countably infinite assumption can be seen as a natural
  approximation of the real world when speaking of contagion or diffusion that 
never stops, does not have any border, but continues and continues. 2) Second, the underlying network is not anymore complete but has some structure, defined by a neighborhood relation with
  very weak requirements, as introduced in \cite{mor00}. The neighborhood
  relation defined by \cite{mor00} amounts to considering as network a connected
  graph where each agent has a bounded number of neighbors, so that almost any kind of
  network can be considered via this neighborhood relation.

  \medskip

  The general aim of the paper is to study the evolution in the long run of the
  diffusion process. More precisely, the fundamental question is: {\it How does
    the diffusion evolve from a finite set of active agents?} This question can
    be decomposed into more basic ones like: Is a cascade
  possible, leading to a state where all agents are active, or all inactive? Is
  polarization possible, i.e., a stable state where a part of the population is
  active while the other is not? Are other remarkable phenomena, like cycles,
  possible? Our main aim is to disentangle which properties are induced by the aggregation function and which properties are induced by the network structure.
	
  Our mechanism of diffusion is Markovian, and we investigate all its absorbing
  and transient classes, which permits to give a precise answer to the above
  questions. However, the assumption on the number of agents (countably
  infinite) makes the number of possible states of the process to be {\it uncountably}
  infinite. This has a strong consequence on the analysis, as it requires the
  use of $\sigma$-fields and much more complex notions than classical ones,
  e.g., irreducibility. 
Our study provides complete results in two cases of aggregation functions:
  strict, and Boolean. A strict aggregation function means that if any agent in
  the neighborhood of agent $i$ becomes active (respectively, inactive), then the probability that $i$ becomes active (respectively, inactive) is positive. A Boolean aggregation function simply means
  that the result is either 0 or 1. Then the model becomes deterministic and we
  are back to the initial model of \cite{mor00}.

There are two classic papers in the prior literature and a number of follow-up works which are particularly relevant, as they already consider the absorbing classes of states for finite graphs and show that the process always converges or the dynamic is two-periodic. These are \cite{gol-oli80} who consider generalized threshold functions, and \cite{pol-sur83} who analyze a weighted majority model. Also quite related is \cite{ber01} who considers a majority model with finite graphs and dynamic monopoly being a subset of nodes whose opinion spreads everywhere; for an overview of dynamic monopolies, see e.g., \cite{pel02}. 
\cite{dre-rob09} consider several graph-theoretical threshold models of the spread of disease or opinion, and focus on an irreversible $k$-threshold process where a node, once getting infected, cannot become uninfected anymore.
There is also a large literature on cascades under general threshold models. One example is 
\cite{gao-gha-sch-yu16} who study stochastic attachment (preferential attachment being a particular case), with a growing network getting newcomers that connect stochastically to nodes already present in the network.
Our model is also related to the extensive line of works on influence maximization originating from \cite{kem-kle-tar03} who analyze mechanisms that are similar to the ones studied in this paper. However, the works along this line usually consider models in which an active node cannot become inactive again. When the random threshold is uniform, we get a particular case of our aggregation model.

  \medskip

  We summarize our main findings. While in the infinite context it is
    typical to consider $\ZZ^2$, we provide more general results that concern
    arbitrary network structures. Most notably, we provide a condition called
    {\it richness} (see below) under which everything essentially works like in
    $\ZZ^2$.  

  In the case of a strict aggregation function, we distinguish two types of
  networks: 
 bipartite and nonbipartite.
  Roughly speaking, a 
  bipartite network amounts to having a partition of the society in two blocks such that neighbors of one block are in the other block.
   Many networks with a regular structure and low degree are bipartite (e.g., the $\ZZ^2$
  plane where each agent has 4 neighbors: north, east, south, west; the
  hexagonal pavement, regular hierarchies, etc.). They are of utmost importance
  in the analysis as they divide the network in two parts, say the {\it even
    agents} and the {\it odd agents}. We show that for nonbipartite networks, 
    there always exists a bipartite subnetwork.
    As a consequence, the results obtained when the network is bipartite carry over to the
  general case with few changes. Our main result gives the exhaustive list of
  all possible absorbing and transient $\phi$-irreducible sets. Most notably, we
  find that: 1) The states where every agent is active or every agent is
  inactive are absorbing, which can result from a cascade effect; 2) Any state
  with a finite number of odd or even agents being active or inactive is
  necessarily transient, so that no polarization can occur, and the state of the
  society evolves towards a mixture of infinitely many active and infinitely
  many inactive agents, either odd or even. $\phi$-irreducibility is proved at
  the price of an additional hypothesis on the structure of the network ({\it
    not} on the aggregation function, interestingly), which we call {\it
    richness}: it says that the network should contain infinitely many ``complex
  stars'' (roughly speaking, stars with at least 3 branches of length at least
  2), and that the network has enough ``space'' to store a local configuration
  of active/inactive agents on its neighbors. Most networks satisfy this
  assumption and we characterize networks having complex stars.

  As for Boolean aggregation functions, the results fairly differ. Indeed, while a
  cascade effect to the limit cases where all agents are active or all are inactive
 is still possible,
  there are many possibilities to reach other absorbing states (which is
  interpreted as polarization), or cycles, depending on the value of the
  threshold. A more detailed study depends essentially on the topology of the
  network, as we showed with $\ZZ^2$ and the neighborhoods of 4 and 8 neighbors.

\medskip

The remainder of this paper is organized as follows. In Section \ref{sec:frame} we introduce the main framework of countable networks and aggregation functions, informal and formal definitions of our diffusion process, and key concepts useful for the analysis of Markov chains with an uncountable set of states. We also present some first results that are independent of specific assumptions on the aggregation functions. 
Section \ref{sec:deterministic} is dedicated to the study of deterministic configurations and the introduction of the notions of a fixed point. 
 In Section \ref{sec:decomp_absorbing_transient} we focus on strict aggregation
 functions and 
bipartite networks.
 We highlight some $\phi$-irreducible
 classes, countable transient and absorbing sets, and we prove in Section
 \ref{sec:decomp_ergodicity} that under some
 additional assumptions on the network, the latter sets are also $\phi$-irreducible.
Section \ref{sec:noblinker} is dedicated to the case of strict aggregation functions 
and nonbipartite networks.
Section \ref{sec:boolean} is devoted to the study of Boolean aggregation functions. 
 We conclude in Section \ref{sec:conclusion}.
Long proofs of our results are presented in Appendices \ref{appA}, \ref{appB}, \ref{appC}, and \ref{appD}.

\section{The framework}\label{sec:frame}

\subsection{Countable networks}\label{subsec:societies}

Let $\scX$ be the set of agents (society) supposed to be countably
infinite. We consider on $\scX$ a {\it neighborhood relation}: $x \sim y$ ($x$
is neighbor of $y$) is a binary relation which is irreflexive, symmetric,
\emph{bounded} (each $x$ has at most $\gamma$ neighbors, where $\gamma$ is a
fixed constant), and \emph{connected}. The latter means that for every $x,y \in
\scX$, there exists a finite path connecting $x$ to $y$, i.e., there exist
$x_1,\ldots,x_k\in \scX$ such that $x_1=x$, $x_k=y$ and for each
$i=1,\ldots,k-1$, $x_i \sim x_{i+1}$. To each agent $x$ we associate its {\it
  neighborhood} $\Gamma^\sim(x)$ as the set of neighbors of $x$ (if no
  confusion occurs, $\Gamma^\sim$ is simply denoted by $\Gamma$). Note that
irreflexivity means that $x\notin \Gamma(x)$. Moreover, we assume that for any
$x$, we have an order on the neighbors in $\Gamma(x)$.

$\scX$ together with the neighborhood relation $\sim$ define an undirected
graph, denoted by $(\scX,\sim)$, whose set of nodes is $\scX$, and there is a
link between $x$ and $y$ iff $x\sim y$. From the above properties of
  $\sim$, the graph is connected, locally finite, and has no self-loop. 
Such a graph, where $\sim$ is a
  neighborhood relation, is called a {\it
    (countable) network}. We give several examples of networks below,
borrowed from \cite{mor00}.

\medskip 

\begin{example}[\textbf{2-dimensional grid}]\label{ex:Z2}
Consider $\scX= \ZZ^2$ and the neighborhood of $x$ is defined as the set
containing all agents within a certain distance of $x$, excluding $x$. Taking
for example the Euclidean distance, the neighborhood of $x$ within distance 1 is
formed by the 4 agents at north, south, east and west position, i.e.,
$(x_1,x_2+1)$, $(x_1,x_2-1)$, $(x_1+1,x_2)$ and $(x_1-1,x_2)$ respectively. We
call this the {\it 1-neighborhood} of $x$. With distance $\sqrt{2}$, we get 8
neighbors (4 on the diagonals in addition to the previous ones). We may consider
more general distances, not necessarily symmetric on both
coordinates. Figure~\ref{fig:Z2} shows the graphs corresponding to the
1-neighborhood and $\sqrt{2}$-neighborhood.
\begin{figure}[htb]
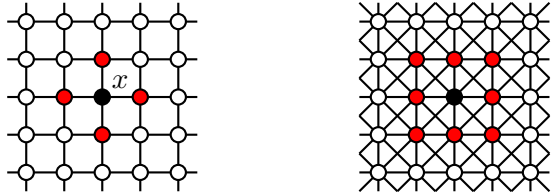

\begin{center}
\psset{unit=0.5cm}
\pspicture(-0.5,-1.5)(4.5,4.5)
\psline(-0.5,0)(4.5,0)
\psline(-0.5,1)(4.5,1)
\psline(-0.5,2)(4.5,2)
\psline(-0.5,3)(4.5,3)
\psline(-0.5,4)(4.5,4)
\psline(0,-0.5)(0,4.5)
\psline(1,-0.5)(1,4.5)
\psline(2,-0.5)(2,4.5)
\psline(3,-0.5)(3,4.5)
\psline(4,-0.5)(4,4.5)
\pscircle[fillstyle=solid](0,0){0.2}
\pscircle[fillstyle=solid](0,1){0.2}
\pscircle[fillstyle=solid](0,2){0.2}
\pscircle[fillstyle=solid](0,3){0.2}
\pscircle[fillstyle=solid](0,4){0.2}
\pscircle[fillstyle=solid](1,0){0.2}
\pscircle[fillstyle=solid](1,1){0.2}
\pscircle[fillstyle=solid,fillcolor=red](1,2){0.2}
\pscircle[fillstyle=solid](1,3){0.2}
\pscircle[fillstyle=solid](1,4){0.2}
\pscircle[fillstyle=solid](2,0){0.2}
\pscircle[fillstyle=solid,fillcolor=red](2,1){0.2}
\pscircle[fillstyle=solid,fillcolor=black](2,2){0.2}
\pscircle[fillstyle=solid,fillcolor=red](2,3){0.2}
\pscircle[fillstyle=solid](2,4){0.2}
\pscircle[fillstyle=solid](3,0){0.2}
\pscircle[fillstyle=solid](3,1){0.2}
\pscircle[fillstyle=solid,fillcolor=red](3,2){0.2}
\pscircle[fillstyle=solid](3,3){0.2}
\pscircle[fillstyle=solid](3,4){0.2}
\pscircle[fillstyle=solid](4,0){0.2}
\pscircle[fillstyle=solid](4,1){0.2}
\pscircle[fillstyle=solid](4,2){0.2}
\pscircle[fillstyle=solid](4,3){0.2}
\pscircle[fillstyle=solid](4,4){0.2}
\uput[45](2,2){$x$}
\endpspicture
\hspace*{2cm}
\psset{unit=0.5cm}
\pspicture(-0.5,-1.5)(4.5,4.5)
\psline(-0.5,0)(4.5,0)
\psline(-0.5,1)(4.5,1)
\psline(-0.5,2)(4.5,2)
\psline(-0.5,3)(4.5,3)
\psline(-0.5,4)(4.5,4)
\psline(0,-0.5)(0,4.5)
\psline(1,-0.5)(1,4.5)
\psline(2,-0.5)(2,4.5)
\psline(3,-0.5)(3,4.5)
\psline(4,-0.5)(4,4.5)
\psline(-0.5,-0.5)(4.5,4.5)
\psline(0.5,-0.5)(4.5,3.5)
\psline(1.5,-0.5)(4.5,2.5)
\psline(2.5,-0.5)(4.5,1.5)
\psline(3.5,-0.5)(4.5,0.5)
\psline(-0.5,0.5)(3.5,4.5)
\psline(-0.5,1.5)(2.5,4.5)
\psline(-0.5,2.5)(1.5,4.5)
\psline(-0.5,3.5)(0.5,4.5)
\psline(-0.5,4.5)(4.5,-0.5)
\psline(0.5,4.5)(4.5,0.5)
\psline(1.5,4.5)(4.5,1.5)
\psline(2.5,4.5)(4.5,2.5)
\psline(3.5,4.5)(4.5,3.5)
\psline(-0.5,3.5)(3.5,-0.5)
\psline(-0.5,2.5)(2.5,-0.5)
\psline(-0.5,1.5)(1.5,-0.5)
\psline(-0.5,0.5)(0.5,-0.5)
\pscircle[fillstyle=solid](0,0){0.2}
\pscircle[fillstyle=solid](0,1){0.2}
\pscircle[fillstyle=solid](0,2){0.2}
\pscircle[fillstyle=solid](0,3){0.2}
\pscircle[fillstyle=solid](0,4){0.2}
\pscircle[fillstyle=solid](1,0){0.2}
\pscircle[fillstyle=solid,fillcolor=red](1,1){0.2}
\pscircle[fillstyle=solid,fillcolor=red](1,2){0.2}
\pscircle[fillstyle=solid,fillcolor=red](1,3){0.2}
\pscircle[fillstyle=solid](1,4){0.2}
\pscircle[fillstyle=solid](2,0){0.2}
\pscircle[fillstyle=solid,fillcolor=red](2,1){0.2}
\pscircle[fillstyle=solid,fillcolor=black](2,2){0.2}
\pscircle[fillstyle=solid,fillcolor=red](2,3){0.2}
\pscircle[fillstyle=solid](2,4){0.2}
\pscircle[fillstyle=solid](3,0){0.2}
\pscircle[fillstyle=solid,fillcolor=red](3,1){0.2}
\pscircle[fillstyle=solid,fillcolor=red](3,2){0.2}
\pscircle[fillstyle=solid,fillcolor=red](3,3){0.2}
\pscircle[fillstyle=solid](3,4){0.2}
\pscircle[fillstyle=solid](4,0){0.2}
\pscircle[fillstyle=solid](4,1){0.2}
\pscircle[fillstyle=solid](4,2){0.2}
\pscircle[fillstyle=solid](4,3){0.2}
\pscircle[fillstyle=solid](4,4){0.2}
\endpspicture
\end{center}
  \caption{Graphs with $\scX=\ZZ^2$ and $\sim$ corresponding to the Euclidean
    distance of 1 (left) and $\sqrt{2}$ (right). The neighborhood of $x$ (black
    node) is the set of red nodes.}
  \label{fig:Z2}
\end{figure}
\end{example}
\begin{example}[\textbf{$d$-dimensional grid}]\label{ex:Zd}
An immediate generalization of the previous example is to consider $\scX=\ZZ^d$,
with $d$ any positive integer, and $\sim$ again defined by the Euclidean
distance. 
\end{example}

\begin{example}[\textbf{hexagonal pavement}]\label{ex:hexp}
The hexagonal pavement is such that every node has 3 neighbors (see
Figure~\ref{fig:hexp}).
\begin{figure}[htb]
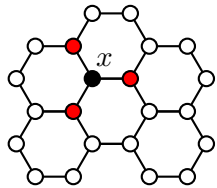

\begin{center}
\psset{unit=0.5cm}
\pspicture(-0.5,-1.5)(4.5,4.5)
\psline(-0.5,0.87)(0,0)(1,0)(1.5,0.87)(2.5,0.87)(3,0)(4,0)(4.5,0.87)
\psline(-0.5,0.87)(0,1.73)(1,1.73)(1.5,0.87)(2.5,0.87)(3,1.73)(4,1.73)(4.5,0.87)
\psline(-0.5,2.6)(0,1.73)(1,1.73)(1.5,2.6)(2.5,2.6)(3,1.73)(4,1.73)(4.5,2.6)
\psline(-0.5,2.6)(0,3.46)(1,3.46)(1.5,2.6)(2.5,2.6)(3,3.46)(4,3.46)(4.5,2.6)
\psline(1,3.46)(1.5,4.33)(2.5,4.33)(3,3.46)
\pscircle[fillstyle=solid](-0.5,0.87){0.2}
\pscircle[fillstyle=solid](0,0){0.2}
\pscircle[fillstyle=solid](1,0){0.2}
\pscircle[fillstyle=solid](1.5,0.87){0.2}
\pscircle[fillstyle=solid](2.5,0.87){0.2}
\pscircle[fillstyle=solid](3,0){0.2}
\pscircle[fillstyle=solid](4,0){0.2}
\pscircle[fillstyle=solid](4.5,0.87){0.2}
\pscircle[fillstyle=solid](0,1.73){0.2}
\pscircle[fillstyle=solid,fillcolor=red](1,1.73){0.2}
\pscircle[fillstyle=solid](3,1.73){0.2}
\pscircle[fillstyle=solid](4,1.73){0.2}
\pscircle[fillstyle=solid](-0.5,2.6){0.2}
\pscircle[fillstyle=solid,fillcolor=black](1.5,2.6){0.2}
\pscircle[fillstyle=solid,fillcolor=red](2.5,2.6){0.2}
\pscircle[fillstyle=solid](4.5,2.6){0.2}
\pscircle[fillstyle=solid](0,3.46){0.2}
\pscircle[fillstyle=solid,fillcolor=red](1,3.46){0.2}
\pscircle[fillstyle=solid](3,3.46){0.2}
\pscircle[fillstyle=solid](4,3.46){0.2}
\pscircle[fillstyle=solid](1.5,4.33){0.2}
\pscircle[fillstyle=solid](2.5,4.33){0.2}
\uput[60](1.5,2.6){$x$}
\endpspicture
\end{center}
 \caption{Graph of the hexagonal pavement}
  \label{fig:hexp}
\end{figure}
\end{example}

\begin{example}[\textbf{hierarchy}]\label{ex:hier}
  A hierarchy is a tree, i.e., a graph where each agent $x$ has $m(x)\geq
  1$ subordinates, and one superior, except the root which has no superior (see
  Figure~\ref{fig:hier}). The numbers $m(x)$ may differ for each agent.
\begin{figure}[htb]
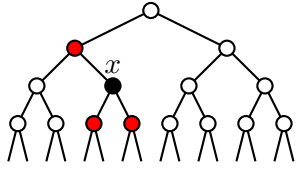

\begin{center}
\psset{unit=0.5cm}
\pspicture(-0.5,-1.5)(7.5,4.5)
\psline(0,0)(0.25,1)(0.5,0)
\psline(1,0)(1.25,1)(1.5,0)
\psline(2,0)(2.25,1)(2.5,0)
\psline(3,0)(3.25,1)(3.5,0)
\psline(4,0)(4.25,1)(4.5,0)
\psline(5,0)(5.25,1)(5.5,0)
\psline(6,0)(6.25,1)(6.5,0)
\psline(7,0)(7.25,1)(7.5,0)
\psline(0.25,1)(0.75,2)(1.25,1)
\psline(2.25,1)(2.75,2)(3.25,1)
\psline(4.25,1)(4.75,2)(5.25,1)
\psline(6.25,1)(6.75,2)(7.25,1)
\psline(0.75,2)(1.75,3)(2.75,2)
\psline(4.75,2)(5.75,3)(6.75,2)
\psline(1.75,3)(3.75,4)(5.75,3)
\pscircle[fillstyle=solid](3.75,4){0.2}
\pscircle[fillstyle=solid,fillcolor=red](1.75,3){0.2}
\pscircle[fillstyle=solid](5.75,3){0.2}
\pscircle[fillstyle=solid](0.75,2){0.2}
\pscircle[fillstyle=solid,fillcolor=black](2.75,2){0.2}
\pscircle[fillstyle=solid](4.75,2){0.2}
\pscircle[fillstyle=solid](6.75,2){0.2}
\pscircle[fillstyle=solid](0.25,1){0.2}
\pscircle[fillstyle=solid](1.25,1){0.2}
\pscircle[fillstyle=solid,fillcolor=red](2.25,1){0.2}
\pscircle[fillstyle=solid,fillcolor=red](3.25,1){0.2}
\pscircle[fillstyle=solid](4.25,1){0.2}
\pscircle[fillstyle=solid](5.25,1){0.2}
\pscircle[fillstyle=solid](6.25,1){0.2}
\pscircle[fillstyle=solid](7.25,1){0.2}
\uput[90](2.75,2){$x$}
\endpspicture
\end{center}
\caption{Graph of a hierarchy with $m(x)=2$ for each agent $x$}
\label{fig:hier}
\end{figure}
\end{example}

\medskip

We assume that each agent can have one of two statuses $0$ and $1$, which can be
interpreted in various ways (opinion on a given subject, adoption of a new
technology, infection by some disease, etc.). We will call this value the
\emph{status} of agent $x$. As we are interested in this paper by the spread
over $\scX$ of one of the actions,
we say that agent $x$
is {\it active} if his status is $1$ and that he is \emph{inactive} if his
status is $0$.

\medskip

At a given time, the society is composed of active and inactive
  agents. The {\it configuration} of the society describes who is active and
  who is inactive. Hence, a configuration can be represented either by a
  0-1-valued function on $\scX$ (0 for inactive, 1 for active), or by the
  set of active agents. The set of all possible configurations is
  $\Omega:=\{0,1\}^\scX\equiv 2^\scX$. In the sequel, we will freely use one or the other
  representation, whichever is most convenient. We elaborate on both
  representations below.

In the function representation, a configuration is a function
$\omega\in\Omega$, where $\omega(x)=1$ if $x$ is active, and $\omega(x)=0$ if
$x$ is inactive. For every $X \subset \scX$, we denote by $\pi_X$ the
projection from $\Omega$ to $\{0,1\}^X$. Every element of $\{0,1\}^{X}$ will be
called a {\it partial configuration} restricted to $X$. Given a subset $X'$, a
\emph{ partial configuration in $X'$} is a partial configuration for a subset
$X$ of $X'$. A partial configuration $\theta \in \{0,1\}^{X}$ defines naturally
the set $\theta^+=\{\omega\in \Omega, \pi_X(\omega)=\theta\}$ of all
configurations compatible with $\theta$. This set will be called the
\emph{cylinder} generated by $\theta$. Two particular configurations are the
constant functions $\1,\0$, which correspond respectively to the cases where all
agents are active and no agent is active. This representation is
particularly convenient for defining formally the dynamical system.

In the set representation, the state of the society is represented by the set
of active agents. We have clearly a one-to-one correspondence between the two
formulations since any configuration $\omega$ is uniquely defined by the set $X$
of active agents, with the correspondence $1_X\equiv\omega$. For example,
$\1,\0$ correspond respectively to $\scX$ and $\varnothing$. A partial
configuration becomes a pair of disjoint sets $(X,Y)$ such that agents in $X$
are active and agents in $Y$ are inactive. It naturally induces
the cylinder $(X,Y)^+$ of configurations such that agents in $X$ are active,
agents in $Y$ are inactive, and the status of agents outside $X\cup Y$ is
not constrained. When $X$ and $Y$ are finite, we will say that the cylinder is
\emph{finite}. This second representation is particularly convenient when using
the notion of neighborhood.

\medskip
\subsection{Informal definition of the diffusion process}

\subsubsection{Basic mechanism}
We now describe how the diffusion evolves in the society. Given a finite set $X$, we denote by $|X|$ its cardinal. Given a natural number
$n$, we denote by $[n]$ the set $\{1,\ldots,n\}$ and for every integer
$k \geq 1$, we denote by $[0,1]^{\leq k}=\cup_{l=1}^k [0,1]^l$,
i.e., the set of all vectors of size at most $k$ whose coordinates
  lie in $[0,1]$.\\

We assume that each agent's status is randomly updated as a function of the
statuses of its neighbors. An {\it aggregation function} is a mapping $A:
[0,1]^{\leq \gamma}\rightarrow[0,1]$ which is nondecreasing w.r.t. each
  coordinate, and satisfies
  $A(1,\ldots,1)=1$, $A(0,\ldots,0)=0$ (recall that $\gamma$ is the
    maximal size of a neighborhood).
It is \textit{symmetric} or \textit{anonymous} if
$A(z_1,\ldots,z_l)=A(z_{\sigma(1)},\ldots,z_{\sigma(l)})$ for every $l\in
[\gamma]$ and every permutation $\sigma$ on $[l]$.\\

\noindent We distinguish 3 cases of aggregation functions:
\begin{enumerate}
\item $A$ is {\it strict}\footnote{Note that if $A$ is increasing (in the strict sense), then it is a strict aggregation function, but not the converse.}: $A(z)=0$ iff $z$ is a $0$ vector and $A(z)=1$ iff $z$ is a $1$ vector.
\item $A$ is {\it Boolean}: $A(z)=0$ or $1$ for all $z\in[0,1]^\gamma$.
\item None of the above applies: $A$ is nonstrict and nonBoolean, i.e., there exist 
 $z\in[0,1]^{\leq\gamma}$ s.t. $0<A(z)<1$, and there exists $z'\neq0$ s.t. $A(z')=0$ or $z''\neq 1$
  s.t. $A(z'')=1$. 
\end{enumerate}

\medskip

We assume that the probability for an agent $x$ to be active at time $t+1$,
given the configuration $\omega$ at time $t$, is
\begin{equation}\label{eq:1}
P(x\mid \omega) = A(\pi_{\Gamma(x)}(\omega)),
\end{equation}
where $A:[0,1]^{\leq \gamma}\rightarrow[0,1]$ is an aggregation function. In
words, the probability for $x$ to be active at the next stage is obtained by
aggregating the vector of statuses of all the neighbors of $x$ in the
configuration $\omega$. By definition, $\pi_{\Gamma(x)}(\omega)$ is a vector in $[0,1]^{\Gamma(x)}$ that we identify to $[0,1]^{|\Gamma(x)|}$ (following the order defined on $\Gamma(x)$), hence the probability is well defined. The more active agents are in the neighborhood, the higher the probability. Moreover, $x$ is active (respectively, inactive) for sure if all its neighbors (respectively, none of its neighbors) were active. We make the additional assumption that the probability for the agents to be active conditionally on $\omega$ is independent across the agents. We will see in the
  next section how to define rigorously this process.

Our framework contains several models presented in the literature. For example,
it may be seen as a sophistication of the voter model \citep{hol-lig75}, where
an agent adopts the status of one of its neighbors, randomly chosen. When
  the aggregation function is Boolean, the updating process becomes
  deterministic, and contains the classical threshold model (see, e.g.,
  \cite{gra78}) as a particular case. This is exactly the framework of
\cite{mor00} and Section~\ref{sec:boolean}. \\

Moreover, as explained in the introduction, \cite{mor00} suggests a game-theoretic interpretation of the threshold-model. The same interpretation can be provided here. Specifically, taking two players $i$ and $j$ with a coordination game payoff matrix yielding $q\in[0,1]$ for coordination on action 0, $1-q$ for coordination on action 1, and 0 otherwise, the best response strategy leads to the choice of action 1 for player $i$ if this player assigns a probability at least $q$ that player $j$ chooses 1. Generalizing this to the neighborhood $\Gamma(i)$ of player $i$, it is found that player $i$'s best response is to choose action 1 if and only if at least $q\gamma$ neighbors choose action 1. Hence, the threshold model is recovered with threshold $q\gamma$. 
In order to get our model in its full generality, it remains to introduce some random device in the game. This can be done for example by proceeding like in the voter model: instead of meeting all of its neighbors, an agent meets only some of them at random, according to some probability distribution. Then the threshold will be exceeded with some probability, which causes the agent to become active (status 1) with some probability.

\subsubsection{Interior and closure of a configuration}
From (\ref{eq:1}) and this informal definition, we can already highlight one
specific aspect of the dynamic process. Using the set representation, we define for any set $X\subseteq \scX$
its \emph{closure} $\clo(X)=\cloX$ and
\emph{interior} $\int(X)=\intX$ by
\begin{align*}
\clo(X)=\cloX &= \{x\in\scX\mymid \Gamma(x)\cap X\neq\varnothing\}\\
\int(X)=\intX &= \{x\in\scX\mymid \Gamma(x)\subseteq X\}.
\end{align*}
It is sometimes convenient to iterate the operator several times. We will define
$\int^n$ (respectively, $\clo^n$) to be the $n$th iteration of the operator $\int$ (respectively, $\clo$). Accordingly, we use the same notation for configurations: $\overline{\omega}, \accentset{\circ}{\omega},
\clo(\omega), \int(\omega)$. Obviously, $\intX\subseteq \cloX$, but it is not true in general that
$\intX\subseteq X\subseteq \cloX$ (see Example~\ref{ex:1} below). Also, it is easy to see that
$\clo$ and $\int$, viewed as mappings on $(2^\scX,\subseteq)$, are monotone:
\[
X\subset X'\Rightarrow \intX\subseteq \accentset{\circ}{X'}\text{ and } \cloX\subseteq \cloX'.
\]
With this new notation, we obtain:
\begin{align}\label{eq:3}
x\in\intX \Rightarrow P(x\mid X)=1, \quad 
x\not\in\cloX\Rightarrow P(x\mid X)=0. 
\end{align}

The above implications become equivalences if and only if $A$ is a strict
aggregation function. This shows that, given that $X$ is the set of active
agents at time $t$, the set $X'$ of active agents at time $t+1$ lies between $\intX$ and $\cloX$. Formally, given $X\subset Z$, denote $[X,Z]:=\{Y\in 2^\scX\mid X\subseteq Y\subseteq Z\}$. Then the set $X'$ of active agents at time $t+1$ lies in the interval $[\intX,\cloX]$ with probability 1.
 Moreover, if the aggregation function is strict and $X$
is finite, then $X'$ can be {\it any} set in $[\intX,\cloX]$ with a positive probability.

\begin{example}\label{ex:1}
Consider $\scX=\ZZ^2$ and the 1-neighborhood. Let $\omega$ be the configuration defined by $\omega(n,m)=1$ iff $|n|+|m|=1$, denoted by $X$ in set representation (Figure~\ref{fig:1} (left)). $\intX$ and $\cloX$ are given in Figure~\ref{fig:1} (middle) and (right), respectively. One can see that $\intX\not\subseteq X\not\subseteq \cloX$, however $\intX\subseteq \cloX$.

\begin{figure}[h!]
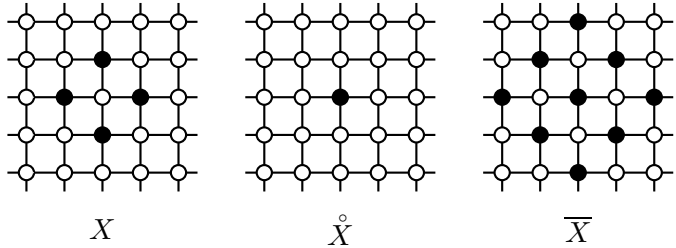

\begin{center}
\psset{unit=0.5cm}
\pspicture(-0.5,-1.5)(4.5,4.5)
\psline(-0.5,0)(4.5,0)
\psline(-0.5,1)(4.5,1)
\psline(-0.5,2)(4.5,2)
\psline(-0.5,3)(4.5,3)
\psline(-0.5,4)(4.5,4)
\psline(0,-0.5)(0,4.5)
\psline(1,-0.5)(1,4.5)
\psline(2,-0.5)(2,4.5)
\psline(3,-0.5)(3,4.5)
\psline(4,-0.5)(4,4.5)
\pscircle[fillstyle=solid](0,0){0.2}
\pscircle[fillstyle=solid](0,1){0.2}
\pscircle[fillstyle=solid](0,2){0.2}
\pscircle[fillstyle=solid](0,3){0.2}
\pscircle[fillstyle=solid](0,4){0.2}
\pscircle[fillstyle=solid](1,0){0.2}
\pscircle[fillstyle=solid](1,1){0.2}
\pscircle[fillstyle=solid,fillcolor=black](1,2){0.2}
\pscircle[fillstyle=solid](1,3){0.2}
\pscircle[fillstyle=solid](1,4){0.2}
\pscircle[fillstyle=solid](2,0){0.2}
\pscircle[fillstyle=solid,fillcolor=black](2,1){0.2}
\pscircle[fillstyle=solid](2,2){0.2}
\pscircle[fillstyle=solid,fillcolor=black](2,3){0.2}
\pscircle[fillstyle=solid](2,4){0.2}
\pscircle[fillstyle=solid](3,0){0.2}
\pscircle[fillstyle=solid](3,1){0.2}
\pscircle[fillstyle=solid,fillcolor=black](3,2){0.2}
\pscircle[fillstyle=solid](3,3){0.2}
\pscircle[fillstyle=solid](3,4){0.2}
\pscircle[fillstyle=solid](4,0){0.2}
\pscircle[fillstyle=solid](4,1){0.2}
\pscircle[fillstyle=solid](4,2){0.2}
\pscircle[fillstyle=solid](4,3){0.2}
\pscircle[fillstyle=solid](4,4){0.2}
\rput(2,-1.5){$X$}
\endpspicture
\hspace*{0.5cm}
\pspicture(-0.5,-1.5)(4.5,4.5)
\psline(-0.5,0)(4.5,0)
\psline(-0.5,1)(4.5,1)
\psline(-0.5,2)(4.5,2)
\psline(-0.5,3)(4.5,3)
\psline(-0.5,4)(4.5,4)
\psline(0,-0.5)(0,4.5)
\psline(1,-0.5)(1,4.5)
\psline(2,-0.5)(2,4.5)
\psline(3,-0.5)(3,4.5)
\psline(4,-0.5)(4,4.5)
\pscircle[fillstyle=solid](0,0){0.2}
\pscircle[fillstyle=solid](0,1){0.2}
\pscircle[fillstyle=solid](0,2){0.2}
\pscircle[fillstyle=solid](0,3){0.2}
\pscircle[fillstyle=solid](0,4){0.2}
\pscircle[fillstyle=solid](1,0){0.2}
\pscircle[fillstyle=solid](1,1){0.2}
\pscircle[fillstyle=solid](1,2){0.2}
\pscircle[fillstyle=solid](1,3){0.2}
\pscircle[fillstyle=solid](1,4){0.2}
\pscircle[fillstyle=solid](2,0){0.2}
\pscircle[fillstyle=solid](2,1){0.2}
\pscircle[fillstyle=solid,fillcolor=black](2,2){0.2}
\pscircle[fillstyle=solid](2,3){0.2}
\pscircle[fillstyle=solid](2,4){0.2}
\pscircle[fillstyle=solid](3,0){0.2}
\pscircle[fillstyle=solid](3,1){0.2}
\pscircle[fillstyle=solid](3,2){0.2}
\pscircle[fillstyle=solid](3,3){0.2}
\pscircle[fillstyle=solid](3,4){0.2}
\pscircle[fillstyle=solid](4,0){0.2}
\pscircle[fillstyle=solid](4,1){0.2}
\pscircle[fillstyle=solid](4,2){0.2}
\pscircle[fillstyle=solid](4,3){0.2}
\pscircle[fillstyle=solid](4,4){0.2}
\rput(2,-1.5){$\intX$}
\endpspicture
\hspace*{0.5cm}
\pspicture(-0.5,-1.5)(4.5,4.5)
\psline(-0.5,0)(4.5,0)
\psline(-0.5,1)(4.5,1)
\psline(-0.5,2)(4.5,2)
\psline(-0.5,3)(4.5,3)
\psline(-0.5,4)(4.5,4)
\psline(0,-0.5)(0,4.5)
\psline(1,-0.5)(1,4.5)
\psline(2,-0.5)(2,4.5)
\psline(3,-0.5)(3,4.5)
\psline(4,-0.5)(4,4.5)
\pscircle[fillstyle=solid](0,0){0.2}
\pscircle[fillstyle=solid](0,1){0.2}
\pscircle[fillstyle=solid,fillcolor=black](0,2){0.2}
\pscircle[fillstyle=solid](0,3){0.2}
\pscircle[fillstyle=solid](0,4){0.2}
\pscircle[fillstyle=solid](1,0){0.2}
\pscircle[fillstyle=solid,fillcolor=black](1,1){0.2}
\pscircle[fillstyle=solid](1,2){0.2}
\pscircle[fillstyle=solid,fillcolor=black](1,3){0.2}
\pscircle[fillstyle=solid](1,4){0.2}
\pscircle[fillstyle=solid,fillcolor=black](2,0){0.2}
\pscircle[fillstyle=solid](2,1){0.2}
\pscircle[fillstyle=solid,fillcolor=black](2,2){0.2}
\pscircle[fillstyle=solid](2,3){0.2}
\pscircle[fillstyle=solid,fillcolor=black](2,4){0.2}
\pscircle[fillstyle=solid](3,0){0.2}
\pscircle[fillstyle=solid,fillcolor=black](3,1){0.2}
\pscircle[fillstyle=solid](3,2){0.2}
\pscircle[fillstyle=solid,fillcolor=black](3,3){0.2}
\pscircle[fillstyle=solid](3,4){0.2}
\pscircle[fillstyle=solid](4,0){0.2}
\pscircle[fillstyle=solid](4,1){0.2}
\pscircle[fillstyle=solid,fillcolor=black](4,2){0.2}
\pscircle[fillstyle=solid](4,3){0.2}
\pscircle[fillstyle=solid](4,4){0.2}
\rput(2,-1.5){$\cloX$}
\endpspicture
\end{center}
\caption{Interior and closure of a set $X$}
\label{fig:1}
\end{figure}
\end{example}

\subsection{Formal definition of the diffusion process}\label{sec:conpro}

We define the diffusion process as a Markov process whose set of
states is the set $\Omega$ of configurations, based on Equation (\ref{eq:1}) and the
independence across agents. As $\Omega$ is uncountable, the definition of the
Markov process requires to work on $\sigma$-fields (see, e.g., \cite{hela03}). \\
%



We consider on $\Omega$ the product topology which is generated by the finite cylinders and we denote by $\mathcal{T}$ the Borelian $\sigma$-field associated to this topology. It is also the product $\sigma$-field.
Let us now define a Markov Kernel on $\Omega$, i.e., a mapping $K$ from $\Omega \times \mathcal{T}$ to $[0,1]$ such that:
\begin{itemize}
\item[-] For every $\omega\in \Omega$, $K(\omega,\cdot)$ is a probability on $\mathcal{T}$,
\item[-] For every $\mathcal{A} \in \cT$, $K(\cdot,\mathcal{A})$ is measurable.
\end{itemize}
$K(\omega,\cA)$ can be interpreted as the probability that from
  configuration $\omega$ the process jumps at next time step into a
  configuration belonging to $\cA$. Similarly, $K^n(\omega,\cA)$ is the
  probability that in $n$ steps the process jumps from $\omega$ into a
  configuration in $\cA$.

 Fix $\omega \in \Omega$. In order to define $K(\omega,\cdot)$, we first
construct a family of probability distributions $(\mu_{Y,\omega})_{Y\subset
  \scX, Y\text{finite}}$ on finite partial configurations.
Consider $h \in \{0,1\}^Y$, we set 
\begin{equation}\label{eq:5}
\mu_{Y,\omega}(\{h\})=\Pi_{y\in Y} \left( P(y\mid\omega)h(y)+ (1-P(y\mid\omega))(1-h(y)) \right),
\end{equation}
where $P(y\mid\omega) = A(\pi_{\Gamma(y)}(\omega))$ (this is
Eq. (\ref{eq:1})). $\mu_{Y,\omega}(\{h\})$ is the probability that, given
  the present configuration $\omega$, the next (partial) configuration in $Y$ is
  $h$. Notice that this probability does not depend on the entire $\omega$ but
  on its projection on $\overline{Y}$. For every $Y$, $\mu_{Y,\omega}$ is a
probability measure on the Borel $\sigma$-field of $Y$ (finite
set). Moreover, this family of probability measures satisfies the
assumption of Kolmogorov Extension Theorem \cite[Corollary 15.27]{ali-bor06}. It
follows that there exists $K(\omega,\cdot)$, a probability distribution over
$\Omega$ and the infinite product $\sigma$-field, for every $\omega\in \Omega$.

\begin{lemma}\label{lem:kernel}
  The function $K(\cdot,\cdot)$ constructed above is a Markov kernel on $\Omega$.
\end{lemma}
(see proof in Appendix \ref{appA})

\medskip

We can now rewrite Equations (\ref{eq:3}) 
in terms of the Markov kernel. Switching to the set notation, let $X$ and $Y$ be two disjoint sets, $(X,Y)^+$ is the
interval $[X,Y^c]$, where $Y^c$ is the complement of $Y$, and with some abuse of notation we keep the same symbol for the $\sigma$-field $\cT$. Given that $X$ is the present configuration, the configuration at next time step lies in the interval $[\intX,\cloX]$ with probability $1$:
\begin{equation}\label{eq:cert}
K(X,\cA)=1 \text{ if } [\intX,\cloX]\subseteq \cA \qquad (X\subseteq \scX).
\end{equation}
Observe that if $X$ is finite, $[\intX,\cloX]$ can easily be seen to be an element of $\cT$ since
\[
[\intX,\cloX]=\bigcap_{x\not\in\cloX}(\intX,\{x\})^+.
\]
The same conclusion holds with infinite sets as well, so that any interval
$[X,Y]$ is an element of $\cT$, and so are the singletons $\{X\}$ for any
$X\subseteq \scX$.  

Formula (\ref{eq:5}) permits to compute the kernel when everything is
finite. Specifically, with $X,Y$ finite, we have
$K(X,\{Y\})=\mu_{\cloX,X}(\{Y\})$. We can also compute the kernel for every
element of the basis. Supposing $\cA=(X,Y)^+$, we have simply
$K(Z,\cA)=\mu_{X\cup Y,Z}(\{X\})$. By $\sigma$-additivity, we can deduce that $K(X,\cA)=0$ if $|X|$ is infinite and $|\cA\cap [\intX,\cloX]|$ is finite.

%
%

\subsection{Properties and first results}

We first introduce some useful notions for the study of Markov chains with an uncountable set of states and present several results that do not depend on additional assumptions on the aggregation function or the neighborhood.


\begin{definition}
A set $\cA \in \mathcal{T}$ is called \emph{absorbing} if $K(\omega,\cA)=1$ for every $\omega \in \cA$. 
\end{definition}
\begin{definition}
A set $\cA\in \mathcal{T}$ is called {\it transient} if for
every configuration $\omega\in \cA$ there exists $n\in\NN$ such that
$K^n(\omega,\cA)<1$. 
\end{definition}
 In words, a set of configurations is absorbing if from any configuration
  in this set, a transition yields a configuration still in this set with
  probability 1, while a set of configurations is transient if there is some
  probability to go outside after a certain number of steps.

  The following result gives two trivial absorbing configurations, which are the
  configurations where every agent is active and where nobody is active.

\begin{proposition}\label{pro:absorbing_partiel}
$\{\textbf{1}\}$ (everybody is active) and $\{\textbf{0}\}$ (nobody is active) are absorbing sets.
\end{proposition}

\begin{definition}
Let $\omega$ be a configuration and $\cA\in\cT$ be a set of configurations. We
say that $\cA$ is {\it reachable} from $\omega$ if after a certain number
  of transitions from $\omega$, we may obtain a configuration in $\cA$:
\[
\sum_{n\in \NN} K^n(\omega,\cA)>0.
\]
\end{definition}

In order to prove that a set is reachable, we will mainly rely on the following
definition of a trajectory. Indeed, the existence of a trajectory from $\omega$
to $\cA$ is equivalent to the fact that $\cA$ is reachable from $\omega$. 
\begin{definition}\label{def:path}
A {\it trajectory} from $\omega$ to $\cA$ is a finite sequence of sets of configurations $(\cA_1,\ldots,\cA_n)$ such that
\begin{itemize}
\item[-] $\cA_1=\{\omega\}$,
\item[-] $\cA_n=\cA$,
\item[-] for every $1\leq l < n$, $\forall m\in \cA_l$, $K(m,\cA_{l+1})>0$.
\end{itemize}
\end{definition}

We will often apply Definition~\ref{def:path} such that every set $\cA_i$ is a cylinder generated by a partial configuration $(X_i,Y_i)$. When doing so, it will be easier to speak directly about a \emph{trajectory of partial configurations}.\\

\begin{definition}
Given $\phi$ a non-trivial $\sigma$-finite measure on $\mathcal{T}$, the Markov
chain is $\phi$-\emph{irreducible} if for all $\cA\in\cT$ such that
$\phi(\cA)>0$, for all $\omega\in \Omega$, 
\[
\sum_{n=1}^{+\infty} K^n(\omega,\cA) >0.
\]
\end{definition}
In words, the Markov chain is $\phi$-irreducible if there is a positive
  probability that starting from any configuration, the process reaches after
  some step any set of configurations, provided this set has a positive measure
  (w.r.t. $\phi$). 

In Proposition \ref{pro:absorbing_partiel} we have seen that the singletons
$\{\0\}$ and $\{\1\}$ are absorbing sets. This implies that the diffusion Markov
process $K$ is not irreducible.
\begin{lemma}
For any non-trivial  $\sigma$-finite measure $\phi$ on $\mathcal{T}$, $K$ is not
$\phi$-irreducible. 
\end{lemma}
\begin{proof}
We prove the result by contradiction. Let $\phi$ be such that the Markov chain is $\phi$-irreducible. Since
$\{\0\}$ is absorbing, for every $\mathcal{A}\in \mathcal{T}$ such that
$\0\notin \mathcal{A}$, we need $\phi(\mathcal{A})=0$. Similarly if $\1\notin
\mathcal{A}$, we need $\phi(\mathcal{A})=0$. $\phi$ is non-trivial, hence there
exists $\cA^*$ such that $\phi(\cA^*) \neq 0$. It follows that $\cA^*$ contains
$\0$ and $\1$. We now consider the following partition of $\cA^*$: $\{\0\}$,
$\{\1\}$ and $\cA^*\setminus\{\0,\1\}$. Then $\phi(\cA^*\setminus\{\0,\1\})=0$ and
  therefore either $\{\0\}$ or $\{\1\}$ has a strictly positive measure under
  $\phi$, which contradicts the fact that a set without $\0$ or $\1$ has a zero
  measure.
\end{proof}

Since $K$ is not $\phi$-irreducible on $\cT$, we may look for sub-fields of
$\cT$ where $K$ is $\phi$-irreducible. We introduce for any $\cA\in\cT$ the {\it
  trace} of $\cT$ on $\cA$, defined by $\mathcal{T}_{|\cA}=\{\cB\in\cT\mymid
\cB\subseteq \cA\}$, and $\phi_{\cA}$ the regular conditional probability on $\cA$.

\begin{definition}
A set of configurations $\cA\in\mathcal{T}$ is a {\it  $\phi$-irreducible set} if for every $\omega \in \cA$, every $\cB\in \mathcal{T}_{|\cA}$ such that $\phi_{\cA}(\cB)>0$, $K^n(\omega,\cB)>0$ for some $n$.
\end{definition}

Combining both $\phi$-irreducibility and the property of being absorbing,
we get the fundamental notion of class.
\begin{definition}
A set of configurations $\cA\in\mathcal{T}$ is a {\it  $\phi$-irreducible class}
if it is both absorbing and a $\phi$-irreducible set.
\end{definition}
This notion is close to the notion of absorbing classes for finite Markov
  chains, and therefore constitutes the most important notion to describe the
  behavior of the Markov process. Indeed, $\phi$-irreducible classes are those
  sets of configurations in which the process stays forever with probability 1,
  and where ``almost'' every subset of configurations is visited.


By Proposition~\ref{pro:absorbing_partiel}, we know that $\{\textbf{1}\}$
  and $\{\textbf{0}\}$ are absorbing singletons. Since a singleton set is
  always $\phi$-irreducible, we obtain our first result on convergence.  
\begin{proposition}\label{prop:01}
$\{\textbf{1}\}$ and $\{\textbf{0}\}$ are finite $\phi$-irreducible classes for any $\phi$.
\end{proposition}
Now that we have found two absorbing singletons and therefore
      having brought a partial answer to our main concern, natural questions
      are: are the singletons $\{\textbf{1}\}$ and $\{\textbf{0}\}$ the only
      absorbing ones?  Does there exist other $\phi$-irreducible classes? We will answer positively  the first question in the next section and study the second one in the subsequent sections.

\section{Deterministic configurations}\label{sec:deterministic}

%

There are two extreme cases where the dynamic process becomes deterministic. The first one is if the aggregation function is Boolean. Then, the transition becomes deterministic independently of the configuration. Given any configuration at stage $t$, only one configuration is possible at stage $t+1$. Informally, the {\it mechanism} makes the process deterministic.
The second one is if the configuration satisfies some special properties. Then, any transition becomes deterministic (starting from this configuration). Informally, the {\it configuration} makes the process deterministic. According to (\ref{eq:cert}), this happens if and only if the configuration $X$ satisfies  $\intX=\cloX$.

Section~\ref{sec:boolean} will be devoted to the case of the deterministic mechanism, while this section is
  devoted to the second case. In order to satisfy $\intX=\cloX$, there are only two possibilities: either
  $\intX =\cloX = X$ (then $X$ is called a \emph{fixed point}), or $\intX =\cloX
  \neq X$. We will see that the latter case leads to a special type of graph.

\begin{proposition}\label{pro:2}
 \begin{enumerate}
    \item $X$ is a fixed point for any aggregation function iff  $X=\scX$ or
      $X=\varnothing$.
    \item If $A$ is strict, then the only fixed points are $X=\scX$ and
      $X=\varnothing$.
   \end{enumerate}
\end{proposition}
\begin{proof}
(1) Let $X$ satisfy $\intX=\cloX=X$ and suppose that $X\neq\varnothing,\scX$. Then
  there exists $x\in X$ and $y\not\in X$. By connectedness, there exists a finite
  path $x_1=x,x_2,\ldots,x_k=y$ connecting $x$ to $y$. However, as $X=\intX$,
  $\Gamma(x_1)\subseteq X$, and consequently $\Gamma(x_2)\subseteq X$,
  $\Gamma(x_3)\subseteq X$, etc., till $\Gamma(x_{k-1})\subseteq X$. However,
  $y=x_k\in\Gamma(x_{k-1})$, a contradiction.

The converse statement is obvious.

  (2) We only have to prove that $X\neq\varnothing,\scX$ cannot be a fixed
  point. The above reasoning can be used without change because when $A$ is
  strict, the implications in (\ref{eq:3}) become equivalences.
\end{proof}

\begin{proposition}\label{pro:3}
Consider $X$ such that
$\intX=\cloX\neq X$. The following holds:
\begin{enumerate}
\item $X^c$ has the same property.
\item $\intX=X^c$ 
\end{enumerate}
\end{proposition}
(see proof in Appendix \ref{appA})

\begin{corollary}\label{cor:1}
There exists a configuration $X$ such that $\intX=\cloX\neq X$ if and only if
the graph $(\scX,\sim)$ is bipartite, with (unique) bipartition $(X,X^c)$. 
\end{corollary}
\begin{proof}
Suppose that $X$ with the above property exists and take $x\in X$. Then
$\Gamma(x)\cap X=\varnothing$. Indeed, if $y\in \Gamma(x)\cap X$, then $y\in
\cloX=\intX=X^c$ by Proposition~\ref{pro:3} (2), a contradiction. The same
property holds for $X^c$ by Proposition~\ref{pro:3} (1). Therefore, $(X,X^c)$ is
a bipartite graph. In addition it is well-known that the bipartition of a
bipartite graph is unique.

The converse statement is obvious. 
\end{proof}

An important consequence of the above result is that if $(\scX,\sim)$ is
bipartite with bipartition $(X,X^c)$, $\{X,X^c\}$ is a periodic absorbing class
of period 2.

Let us revisit the networks given in Examples~\ref{ex:Z2} to \ref{ex:hier}
and check if they are bipartite.  $\ZZ^2$ with the 1-neighborhood is bipartite,
with the bipartition given by the parity (odd, even) of the nodes: a node
$(n,m)$ is called ``even'' (resp., odd) if $n+m$ is even (resp., odd) (see
Figure~\ref{fig:blink} (left)). However, $\ZZ^2$ is no more bipartite with the
$\sqrt{2}$-neighborhood, nor with any other distance. The result can be extended
to $\ZZ^d$ as well, for any positive integer $d$. The hexagonal pavement is also
bipartite (1 node out of 2 on each hexagon; see Figure~\ref{fig:blink} (middle))
and so is the hierarchy (all nodes of odd-numbered layers; see
Figure~\ref{fig:blink} (right)). 

\begin{figure}[htb]
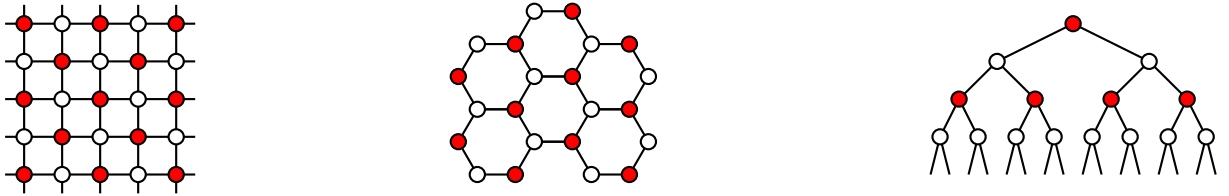

\begin{center}
\psset{unit=0.5cm}
\pspicture(-0.5,-1.5)(4.5,4.5)
\psline(-0.5,0)(4.5,0)
\psline(-0.5,1)(4.5,1)
\psline(-0.5,2)(4.5,2)
\psline(-0.5,3)(4.5,3)
\psline(-0.5,4)(4.5,4)
\psline(0,-0.5)(0,4.5)
\psline(1,-0.5)(1,4.5)
\psline(2,-0.5)(2,4.5)
\psline(3,-0.5)(3,4.5)
\psline(4,-0.5)(4,4.5)
\pscircle[fillstyle=solid,fillcolor=red](0,0){0.2}
\pscircle[fillstyle=solid](0,1){0.2}
\pscircle[fillstyle=solid,fillcolor=red](0,2){0.2}
\pscircle[fillstyle=solid](0,3){0.2}
\pscircle[fillstyle=solid,fillcolor=red](0,4){0.2}
\pscircle[fillstyle=solid](1,0){0.2}
\pscircle[fillstyle=solid,fillcolor=red](1,1){0.2}
\pscircle[fillstyle=solid](1,2){0.2}
\pscircle[fillstyle=solid,fillcolor=red](1,3){0.2}
\pscircle[fillstyle=solid](1,4){0.2}
\pscircle[fillstyle=solid,fillcolor=red](2,0){0.2}
\pscircle[fillstyle=solid](2,1){0.2}
\pscircle[fillstyle=solid,fillcolor=red](2,2){0.2}
\pscircle[fillstyle=solid](2,3){0.2}
\pscircle[fillstyle=solid,fillcolor=red](2,4){0.2}
\pscircle[fillstyle=solid](3,0){0.2}
\pscircle[fillstyle=solid,fillcolor=red](3,1){0.2}
\pscircle[fillstyle=solid](3,2){0.2}
\pscircle[fillstyle=solid,fillcolor=red](3,3){0.2}
\pscircle[fillstyle=solid](3,4){0.2}
\pscircle[fillstyle=solid,fillcolor=red](4,0){0.2}
\pscircle[fillstyle=solid](4,1){0.2}
\pscircle[fillstyle=solid,fillcolor=red](4,2){0.2}
\pscircle[fillstyle=solid](4,3){0.2}
\pscircle[fillstyle=solid,fillcolor=red](4,4){0.2}
\endpspicture
\hfill
\pspicture(-0.5,-1.5)(4.5,4.5)
\psline(-0.5,0.87)(0,0)(1,0)(1.5,0.87)(2.5,0.87)(3,0)(4,0)(4.5,0.87)
\psline(-0.5,0.87)(0,1.73)(1,1.73)(1.5,0.87)(2.5,0.87)(3,1.73)(4,1.73)(4.5,0.87)
\psline(-0.5,2.6)(0,1.73)(1,1.73)(1.5,2.6)(2.5,2.6)(3,1.73)(4,1.73)(4.5,2.6)
\psline(-0.5,2.6)(0,3.46)(1,3.46)(1.5,2.6)(2.5,2.6)(3,3.46)(4,3.46)(4.5,2.6)
\psline(1,3.46)(1.5,4.33)(2.5,4.33)(3,3.46)
\pscircle[fillstyle=solid,fillcolor=red](-0.5,0.87){0.2}
\pscircle[fillstyle=solid](0,0){0.2}
\pscircle[fillstyle=solid,fillcolor=red](1,0){0.2}
\pscircle[fillstyle=solid](1.5,0.87){0.2}
\pscircle[fillstyle=solid,fillcolor=red](2.5,0.87){0.2}
\pscircle[fillstyle=solid](3,0){0.2}
\pscircle[fillstyle=solid,fillcolor=red](4,0){0.2}
\pscircle[fillstyle=solid](4.5,0.87){0.2}
\pscircle[fillstyle=solid](0,1.73){0.2}
\pscircle[fillstyle=solid,fillcolor=red](1,1.73){0.2}
\pscircle[fillstyle=solid](3,1.73){0.2}
\pscircle[fillstyle=solid,fillcolor=red](4,1.73){0.2}
\pscircle[fillstyle=solid,fillcolor=red](-0.5,2.6){0.2}
\pscircle[fillstyle=solid](1.5,2.6){0.2}
\pscircle[fillstyle=solid,fillcolor=red](2.5,2.6){0.2}
\pscircle[fillstyle=solid](4.5,2.6){0.2}
\pscircle[fillstyle=solid](0,3.46){0.2}
\pscircle[fillstyle=solid,fillcolor=red](1,3.46){0.2}
\pscircle[fillstyle=solid](3,3.46){0.2}
\pscircle[fillstyle=solid,fillcolor=red](4,3.46){0.2}
\pscircle[fillstyle=solid](1.5,4.33){0.2}
\pscircle[fillstyle=solid,fillcolor=red](2.5,4.33){0.2}
\endpspicture
\hfill
\pspicture(-0.5,-1.5)(7.5,4.5)
\psline(0,0)(0.25,1)(0.5,0)
\psline(1,0)(1.25,1)(1.5,0)
\psline(2,0)(2.25,1)(2.5,0)
\psline(3,0)(3.25,1)(3.5,0)
\psline(4,0)(4.25,1)(4.5,0)
\psline(5,0)(5.25,1)(5.5,0)
\psline(6,0)(6.25,1)(6.5,0)
\psline(7,0)(7.25,1)(7.5,0)
\psline(0.25,1)(0.75,2)(1.25,1)
\psline(2.25,1)(2.75,2)(3.25,1)
\psline(4.25,1)(4.75,2)(5.25,1)
\psline(6.25,1)(6.75,2)(7.25,1)
\psline(0.75,2)(1.75,3)(2.75,2)
\psline(4.75,2)(5.75,3)(6.75,2)
\psline(1.75,3)(3.75,4)(5.75,3)
\pscircle[fillstyle=solid,fillcolor=red](3.75,4){0.2}
\pscircle[fillstyle=solid](1.75,3){0.2}
\pscircle[fillstyle=solid](5.75,3){0.2}
\pscircle[fillstyle=solid,fillcolor=red](0.75,2){0.2}
\pscircle[fillstyle=solid,fillcolor=red](2.75,2){0.2}
\pscircle[fillstyle=solid,fillcolor=red](4.75,2){0.2}
\pscircle[fillstyle=solid,fillcolor=red](6.75,2){0.2}
\pscircle[fillstyle=solid](0.25,1){0.2}
\pscircle[fillstyle=solid](1.25,1){0.2}
\pscircle[fillstyle=solid](2.25,1){0.2}
\pscircle[fillstyle=solid](3.25,1){0.2}
\pscircle[fillstyle=solid](4.25,1){0.2}
\pscircle[fillstyle=solid](5.25,1){0.2}
\pscircle[fillstyle=solid](6.25,1){0.2}
\pscircle[fillstyle=solid](7.25,1){0.2}
\endpspicture
\end{center}
  \caption{Bipartite graphs: $\ZZ^2$ endowed with the 1-neighborhood (left), the
    hexagonal pavement (middle), and the hierarchy (right)}
  \label{fig:blink}
\end{figure}

\medskip

We comment on the results obtained so far.
\begin{itemize}
\item We have found that, when $A$ is strict, the only
  possible stable configurations are the extreme ones $\scX$ and $\varnothing$,
  where agents are either all active or all inactive. This means that there
  cannot be polarization {\it stricto sensu}, i.e., a stable situation where
  forever a set of agents will be inactive while the other ones are active. We
  will see in Section~\ref{sec:boolean} that with Boolean aggregation functions,
  other fixed points exist.
\item If the network is a bipartite graph, and only under this condition, a second case
  of deterministic configuration exists, which is a cycle of length 2, where
  alternatively half of the agents are active, then inactive at the next time
  step and so on. This is independent of the aggregation function.
  Although such a cycle has almost no chance to occur in
  practice, it will reveal to be a fundamental tool in our analysis of
  convergence. What is more, we will show that, even when the network is not
  bipartite, there always exists a subnetwork which is bipartite, so that
  the analysis amounts to the same. These important facts lead to the structure
  of the rest of the paper: we conduct a full analysis of convergence for
  bipartite networks (Sections~\ref{sec:decomp_absorbing_transient} and \ref{sec:decomp_ergodicity}), then using these results as a
  basis, we extend them to the general case (Section~\ref{sec:noblinker}).
\item The result on the cycle of length 2 is reminiscent of the classical result
  on finite graphs with the threshold model (recall that this is a particular
  case of our Boolean aggregation function model), saying that the process
  either reaches a fixed point or enters a cycle of length 2 (see 
  \cite{gol-oli80}, \cite{pol-sur83}). We mention that this classical result is no
  more valid for infinite graphs as shown by 
  \cite{mor95} where we study
  the case of deterministic mechanism, we will see other examples of cycles.
\end{itemize}


From now on, the bipartition of a bipartite network $(\scX,\sim)$ will be
denoted by $\{\scX_e,\scX_o\}$, where $\scX_e$ is the set of {\it even
  positions} or nodes, and $\scX_o$ the set of {\it odd positions}, by analogy
with the example of $\ZZ^d$ and the 1-neighborhood. It
is important to keep in mind that any $x\in \scX_e$ (resp., $\scX_o$) has its
neighborhood contained in $\scX_o$ (resp., $\scX_e$). In particular, the status
at even stages at even positions is independent from the status at odd stages at
even positions.

\section{Strict aggregation functions and bipartite networks: Absorbing and transient sets}\label{sec:decomp_absorbing_transient}
In this section, we focus on the case of strict aggregation functions and bipartite
networks $(\scX,\sim)$. A key element of the analysis is to
define a partition of the set of configurations into blocks, according to the
``number'' of active/inactive nodes on odd and even positions.

\medskip

A block of the partition of $\Omega$ is characterized by a 4-uple
$(E_a,E_i,O_a,O_i)$, where $E_a,E_i,O_a,O_i$ are respectively the ``numbers'' of
nodes in $\scX_e$ (even positions) with inactive status, in $\scX_e$ with active
status, in $\scX_o$ with inactive status , and in $\scX_o$ with active
status. The ``numbers'' of nodes are coded by $0,\F,\infty$, where 0 means that
there is no node, $\F$ means that the number of nodes is finite but positive,
and $\infty$ means that there are infinitely many nodes.

There are a priori $3^4 = 81$ possible blocks in the partition, but actually
only 25 are nonempty. Indeed, if we restrict ourselves to even positions, either
0 or 1 has to appear infinitely often, therefore only 5 possible cases remain: $(0,\infty)$, $(\F,\infty)$, $(\infty,0)$, $(\infty,\F)$ and $(\infty,\infty)$. Since what happens on odd and on even positions are independent, there are $25$ nonempty blocks in the partition.

For example, with this notation we have:
\begin{itemize}
\item[-] $(\infty,0,\infty,0)$ is the singleton $\{\textbf{0}\}=\{\varnothing\}$,
\item[-] $(0,\infty, 0, \infty)$ is the singleton $\{\textbf{1}\}=\{\scX\}$,
\item[-] $(0, \infty, \infty, 0)$ is the singleton $\{\scX_e\}$ which is the configuration where even agents are active whereas odd agents are inactive,
\item[-] $ (\infty,0,0,\infty)$ is the singleton $\{\scX_o\}$ which is the configuration where even agents are inactive whereas odd agents are active.
\end{itemize}

We now establish the following theorem. 

\begin{theorem}\label{absorbing-transient}
The following sets of configurations are respectively:
\begin{enumerate}[(i)]
\item \label{positive} finite $\phi$-irreducible classes:
\begin{itemize}
\item[-] $(\infty,0,\infty,0)$,
\item[-] $(0,\infty,0,\infty)$,
\item[-] $(0,\infty,\infty, 0) \cup (\infty, 0, 0,\infty)$.
\end{itemize}
\item \label{null} infinite uncountable absorbing sets: 
\begin{itemize}
\item[-] $(\infty,\infty,\infty,\infty)$,
\item[-] $(\infty, 0 ,\infty,\infty)\cup (\infty,\infty,\infty,0)$,
\item[-] $(0, \infty ,\infty,\infty)\cup (\infty,\infty,0,\infty)$.
\end{itemize}
\item \label{transient} infinite transient sets: 
\begin{multicols}{2}
\begin{itemize}
\item[-] $(\infty,\F,\infty,\F)$,
\item[-] $(\F,\infty,\F,\infty)$,
\item[-] $(\F,\infty,\infty,\F) \cup (\infty,\F,\F,\infty)$, 
\item[-] $(\infty, \F ,\infty,\infty)\cup (\infty,\infty,\infty,\F)$,
\item[-] $(\F, \infty ,\infty,\infty)\cup (\infty,\infty,\F,\infty)$,
\item[-] $(\infty,\F,\infty,0) \cup (\infty,0,\infty,\F)$,
\item[-] $(\F,\infty,0,\infty) \cup(0,\infty,\F,\infty)$,
\item[-] $(0,\infty,\infty,\F) \cup (\infty,\F,0,\infty)$, 
\item[-] $(\F,\infty,\infty,0) \cup (\infty,0,\F,\infty)$.
\end{itemize}
\end{multicols}
\end{enumerate}
\end{theorem}
(see proof in Appendix \ref{appB})

Theorem \ref{absorbing-transient}.(\ref{positive}) reformulates and synthetizes
previous results (Propositions~\ref{prop:01},~\ref{pro:2} and \ref{pro:3})
thanks to this new notation, on which we have already commented. Results
  (\ref{null}) and (\ref{transient}), although important and interesting, are
  not conclusive for our study of convergence. Indeed, we have not yet proved
  that these sets are irreducible, therefore we cannot conclude that these sets
  are classes. We postpone their interpretation till
  Theorem~\ref{irreducibility}, where this result will be established.

\section{Strict aggregation functions and bipartite networks:
  $\phi$-irreducible sets}\label{sec:decomp_ergodicity}

The next step is to study in detail $\phi$-irreducibility. Contrarily to the previous
results which hold without additional assumptions on the network, we need
now some particular conditions to establish $\phi$-irreducibility. In this
section, based on necessary conditions on $(\scX,\sim)$, we will prove that the sets highlighted in Theorem \ref{absorbing-transient} are $\phi$-irreducible. Hence, we provide a
decomposition of the set of configurations into three types: sets in case (\ref{positive}) are finite absorbing $\phi$-irreducible sets, those in case (\ref{null}) are infinite absorbing $\phi$-irreducible sets, and the ones in case (\ref{transient}) are transient
$\phi$-irreducible sets.

In the remainder, we may consider on $\Omega$ the following measure $\phi$ such that for every two finite disjoint sets $X$ and $Y$,
\[
\phi((X,Y)^+)=\frac{1}{2^{|X|+|Y|}}.
\]
For every $X^*$ and $Y^*$, we impose the restriction of $\phi$ conditionally on $(X^*,Y^*)^+$ defined on the trace of $\mathcal{T}$ on $(X^*,Y^*)^+$ such that for every two finite disjoint sets $X$ and $Y$ disjoint with $X^* \cup Y^*$,
\[
\phi_{X^*,Y^*}((X^* \cup X, Y^* \cup Y)^+)=\frac{1}{2^{|X|+|Y|}}.
\]

In general, the sets presented in the previous section are not
$\phi$-irreducible for any $(\scX,\sim)$. We
present two counterexamples in order to highlight two different types of
problem.

\begin{example}\label{counter_line}
Consider the network defined by $\scX=\ZZ$ and the 1-neighborhood:
\[
\forall x\in \scX, \Gamma(x)=\{x-1,x+1\}.
\]
We want to show that the set of configurations
$(\infty,\infty,\infty,0)\cup(\infty,0,\infty,\infty)$ is not
$\phi$-irreducible. Consider the configuration $\omega$ in this set defined by
\[
\omega(x)=\begin{cases} 0 \text{ if } x \text{ is odd},\\
0 \text{ if } x\leq 0 \text{ and } x \text{ is even},\\
 1 \text{ if } x>0  \text{ and } x \text{ is even}.
 \end{cases}
\]
and the cylinder $\cA=(X,Y)^+$ with $X=\{0,4\}$ and $Y=\{2\}$. We show that the
cylinder $\cA$ is not reachable from $\omega$. The reason is that at any stage, the current configuration $\omega'$ has the form $\omega'(x)=\omega(x-z)$, for some $z\in \ZZ$, so that the succession 1-0-1 (when looking only at even positions in an even stage or only at odd positions in an odd stage) is not possible. Let us show this for the first transition from $\omega$. The new configuration $\omega'$ is such that $\omega'(2k+1)=1$ for any $k\geq 1$, $\omega'(1)$ can be either 1 or 0, and $\omega'(x)=0$ otherwise. If  $\omega'(1)=0$, then the claim is true with $z=1$, and if $\omega'(1)=1$, then the claim is true with $z=-1$. As $\omega'$ is equal to $\omega$ up to a shift, the same reasoning applies at every stage.

\end{example}

\begin{example}\label{counter_collision}
Consider the plane graph completed by two additional agents linked to $(0,0)$. Formally, let $\scX=\ZZ^2 \cup \{\alpha,\beta\}$. We assume that the neighborhoods are given by
\begin{itemize}
\item $\Gamma(\alpha)=\Gamma(\beta)=\{(0,0)\}$,
\item $\Gamma((0,0))=\{(-1,0),(0,-1),(1,0),(0,1),\alpha,\beta\}$,
\item for every $(n,m)\neq (0,0)$, $\Gamma((n,m))=\{(n-1,m),(n,m-1),(n+1,m),(n,m+1)\}$.
\end{itemize}
Then, this network is bipartite with $\scX_e=\{(n,m) \text{ s.t. } n+m \text{
  is even}\}$ and $\scX_o=\{(n,m) \text{ s.t. } n+m \text{ is odd}\}\cup \{\alpha,\beta\}$. Moreover, since $\alpha$ and $\beta$ have a unique neighbor, we know that for every dynamic of opinion: the status of $\alpha$ (resp. $\beta$) at stage $t$ is the status of agent $(0,0)$ at stage $t-1$. In particular, both statuses are equal and it is impossible to reach the cylinder $(\{\alpha\},\{\beta\})^+$ of configurations where $\alpha$ is active whereas $\beta$ is inactive.
\end{example}

Intuitively, in Example~\ref{counter_line}, the set of configurations is not
$\phi$-irreducible because the graph is
  a line, and there is not enough room to move 0 and 1 without erasing patterns
  that we want to preserve. In Example~\ref{counter_collision}, the graph is
  enough large to move patterns, but as the nodes $\alpha,\beta$ share the same
  neighbor, there is no possibility to store somewhere a partial configuration
  where $\alpha$ and $\beta$ would take different values.\\

Based on these examples, in the next section we derive sufficient
  conditions on graphs in order to obtain $\phi$-irreducibility.

\subsection{Conditions on graphs}

The first definition is related to the problem encountered in Example \ref{counter_line}.

\begin{definition}
A {\it complex star} is a 7-uple $(s_*,s_1,s_2,s_3,s'_1,s'_2,s'_3)\in\scX^7$ such that:
\begin{itemize}
\item[-] $s_1$, $s_2$, $s_3$ are 3 distinct nodes;
\item[-] $\{s_1,s_2,s_3\}\subseteq\Gamma(s_*)$;
\item[-] $s'_1 \in \Gamma(s_1)$, $s'_2 \in \Gamma(s_2)$ and $s'_3 \in \Gamma(s_3)$;
\item[-] $s_* \notin \{s'_1,s'_2,s'_3\}$.
\end{itemize}
Informally, $s_*$ is the center of a star with three branches that have at least
a depth of two. Recall that by assumption $(\scX,\sim)$ is bipartite, and therefore
we know that $\{s'_1,s'_2,s'_3,s_*\}\cap \{s_1,s_2,s_3\}=\varnothing$. Also,
note that we do not assume that $s'_1$, $s'_2$ and $s'_3$ are distinct.
\end{definition}

Roughly speaking, a complex star is a device permitting to store an inactive
status or an active status along time. Its existence prevents the graph to be
similar to a line. In the following proposition, which characterizes
  networks having no complex stars, we identify a sequence with the set of its
values.

\begin{proposition}\label{no_star_graph}
The network $(\scX,\sim)$ does not contain a complex star if and only if it has
one of the two following forms:
\begin{itemize}
\item[-] First case:
\begin{enumerate}
\item there exists a sequence $(x_m)_{m\in \NN}$ in $\scX^{\NN}$ such that all
  nodes are different, $\Gamma(x_0)=\{x_1\} $ and for every $n\geq 1$,
 $\Gamma(x_n)\cap (x_m)_{m\in \NN}=\{x_{n-1},x_{n+1}\}$.
\item for all $x\in \scX$, if $x\notin (x_m)_{m\in \NN}$ then there exists $n\in \NN$ such that $\Gamma(x)=\{x_n\}$.
\end{enumerate}
\item[-] Second case:
\begin{enumerate}
\item there exists a sequence $(x_m)_{m\in \ZZ}$ in $\scX^{\ZZ}$ such that all
  nodes are different and for every $n \in \ZZ$, $
  \Gamma(x_n)\cap (x_m)_{m\in \ZZ}=\{x_{n-1},x_{n+1}\}$.
\item for all $x\in \scX$, if $x\notin (x_m)_{m\in \ZZ}$ then there exists $n\in \ZZ$ such that $\Gamma(x)=\{x_n\}$.
\end{enumerate}
\end{itemize}
\end{proposition}
(see proof in Appendix \ref{appA})

 The existence of a complex star is not sufficient to ensure $\phi$-irreducibility, as it is shown by the network of Example~\ref{counter_collision}. As explained above, the failure is due to the fact that the desired partial configuration cannot be stored on the unique neighbor of $\alpha,\beta$. This motivates the following definition. A similar definition can be stated on $\scX_e$ or on $\scX$.

\begin{definition}\label{saving}
We say that a partial configuration $(X,Y)$ of $\mathcal{X}_o$  can be {\it stored} if there exists a mapping $\theta$ from $\mathcal{X}_o$ to $\mathcal{X}_e$ such that 
\begin{enumerate}
\item for every $x\in X\cup Y$, $\theta(x)\in \Gamma(x)$,
\item for every $x\in X$ and every $y \in Y$, $\theta(x) \neq \theta(y)$.
\end{enumerate}
$\theta$ is called a \emph{storing function}. Observe that $\theta(X\cup
Y)\subseteq \overline{X\cup Y}$.
\end{definition}
It is easy to check that if a partial configuration $(X_1,Y_1)$ of
$\mathcal{X}_o$ can be stored and a partial configuration $(X_2,Y_2)$ of
$\mathcal{X}_e$ can be stored then the partial configuration $(X_1\cup
X_2,Y_1\cup Y_2)$ can also be stored. Moreover, if any partial configuration can
be stored, then any configuration can be stored as well.



We now present classes of graphs where a storing function exists. We focus on
injective storing functions which permit an easy analysis. Consider a partial
configuration $(X,Y)$ in $\scX_o$, storable by $\theta$ and suppose in addition
that $\theta$ is an injection from $X\cup Y$ to $\overline{X\cup Y}\subseteq
\scX_e$. In terms of graph theory, we may rephrase this by saying that the
bipartite graph $(X\cup Y,\overline{X\cup
  Y},\sim)$ admits a {\it matching of $X\cup Y$}.

There are classes of graphs $(\scX,\sim)$ which admit a matching of $X\cup Y$ for any
partial configuration $(X,Y)$ in $\scX_o$ and in $\scX_e$: these are the
$k$-regular graphs and the hierarchies (see Example~\ref{ex:hier}). A graph is
{\it $k$-regular} ($k\geq 2$, integer) if each node has exactly $k$ neighbors.
\begin{proposition}\label{prop:kreg}
If $(\scX,\sim)$ is $k$-regular or is a hierarchy, then any configuration can be stored.
\end{proposition}

(see proof in Appendix \ref{appA})

This leads to the following condition on the network.
\begin{quote}
 {\bf Richness Assumption.} $(\scX,\sim)$ is said to be \emph{rich} if: 
\begin{enumerate}
\item There exists an infinite number of complex stars;
\item Any partial configuration $(X,Y)$ on $\mathcal{X}_o$ and on $\mathcal{X}_e$ can be stored.
\end{enumerate}
\end{quote}
Observe that all the networks introduced in Examples~\ref{ex:Z2} to
\ref{ex:hier} satisfy the richness assumption (except $(\ZZ^d,\sim)$ with $d=1$
which has no complex star), since they are all $k$-regular or a hierarchy, and
they contain infinitely many complex stars.

\begin{remark}
We already described several particular classes of graphs that admit perfect matching. Let us take another point of view. How likely is it for a random bipartite graph to admit a perfect matching? \cite{erd-ren66} showed that for finite random graphs, if there exists a connected component, then it is essentially unique and with probability close to 1 there exists a perfect matching. Hence, in terms of random graphs, connectivity essentially implies that any configuration can be stored.
\end{remark}

\subsection{Main result}

We can now state our result on $\phi$-reachability and irreducibility.
\begin{theorem}\label{irreducibility}
Assume that $(\scX,\sim)$ satisfies the Richness Assumption. Then, each set of Theorem \ref{absorbing-transient} is $\phi$-irreducible.
\end{theorem}
(see proof in Appendix \ref{appC})

Combining Theorems~\ref{absorbing-transient} and \ref{irreducibility}, we
  immediately obtain our final result on convergence.
  \begin{corollary}\label{cor:blk}
    Assume that $(\scX,\sim)$ satisfies the Richness Assumption. The following
    sets are
\begin{enumerate}[(i)]
\item  Finite $\phi$-irreducible classes:
\begin{itemize}
\item[\bf (1)]  $(0,\infty,0,\infty)$ (this is $\scX$);
\item[(2)] $(\infty,0,\infty,0)$ (this is $\varnothing$);
\item[\bf (3)] $(0,\infty,\infty, 0) \cup (\infty, 0, 0,\infty)$ (this is the
 2-cycle on the bipartition);
\end{itemize}
\item Infinite (uncountable) $\phi$-irreducible classes:
\begin{itemize}
\item[\bf (4)] $(\infty,\infty,\infty,\infty)$;
\item[\bf (5)] $(\infty, 0 ,\infty,\infty)\cup (\infty,\infty,\infty,0)$;
\item[(6)] $(0, \infty ,\infty,\infty)\cup (\infty,\infty,0,\infty)$;
\end{itemize}
\item Transient and $\phi$-irreducible sets: 
\begin{multicols}{2}
\begin{itemize}
\item[\bf (a)] $(\infty,\F,\infty,\F)$;
\item[(b)] $(\F,\infty,\F,\infty)$;
\item[\bf (c)] $(\F,\infty,\infty,\F) \cup (\infty,\F,\F,\infty)$; 
\item[\bf (d)] $(\infty, \F ,\infty,\infty)\cup (\infty,\infty,\infty,\F)$;
\item[(e)] $(\F, \infty ,\infty,\infty)\cup (\infty,\infty,\F,\infty)$;
\item[\bf (f)] $(\infty,\F,\infty,0) \cup (\infty,0,\infty,\F)$;
\item[(g)] $(\F,\infty,0,\infty) \cup(0,\infty,\F,\infty)$;
\item[\bf (h)] $(0,\infty,\infty,\F) \cup (\infty,\F,0,\infty)$; 
\item[(i)] $(\F,\infty,\infty,0) \cup (\infty,0,\F,\infty)$.
\end{itemize}
\end{multicols}
\end{enumerate}
  \end{corollary}
\noindent  We make some comments on these results.
  \begin{itemize}
  \item We have found 3 finite classes, 3 infinite (uncountable) classes and 9
    transient sets. We observe that these numbers can be reduced taking into
    account symmetries. Indeed, the collection of classes and the collection of
    transient sets are invariant by interchanging ``active'' and ``inactive''
    (coordinates 1 and 2, and 3 and 4). This yields respectively 2 types of
    finite classes ((1) and (3)), 2 types of infinite classes ((4) and (5)), and
    5 types of transient sets ((a), (c), (d), (f) and (h)) (we have put these
    labels in bold). It is therefore
    enough to restrict our analysis to these classes and sets.

    We also observe that every class and every transient set is invariant
    to interchanging ``odd'' and ``even'' (that is, coordinates 1 and 3, and
    2 and 4).
  \item A striking fact is that every set of configurations with a finite number
    of active or inactive nodes is transient, which means that any such
    configuration is doomed to disappear, sooner or later. This means that there
    cannot be polarization of any type in the long run: a finite set of agents
    (this can be called {\it group}) cannot remain active or inactive forever,
    nor any cycle or periodic phenomenon involving several groups, nor any kind
    of weak version of groups (e.g., finite sets with some fluctuation) can
    exist.
  \item Recall that our fundamental question was how the diffusion
    evolves from a finite set of active agents. We have partly answered above:
    no such configuration can last very long as finite sets of active/inactive
    agents fade out. Any initial configuration must evolve and finish in one of
    the classes we have enumerated. But in which classes may finish a given
    initial configuration with a finite set of active or inactive agents? We do
    not have a general result specifying this, but it is not difficult to find,
    at least with $\ZZ^2$ and the 1-neighborhood, finite sets of active agents
    which can evolve in any desired class from (1) to (6). Hence, we conjecture that in general,
    any class can be reached. However, saying with which
    probability appears to be extremely complex.
  \item Among the classes we have 2 singletons ((1) and (2)), one cycle of
    length 2 ((3)), two periodic classes of length 2 ((5) and (6)), and one
    aperiodic class ((4)). Classes (5) and (6) are periodic, because taking for
    instance (5) and starting from a configuration in $(\infty,0,\infty,\infty)$,
    we have at next stage a configuration in $(\infty,\infty,\infty,0)$, and
    conversely. This is because if there are no active agents on even nodes, all
    odd nodes have all their neighbors inactive, so all the odd nodes must
    become inactive at next stage. This is similar to a 2-cycle, up to the
    difference that the status of agents of one block (odd or even,
    alternatively) is not determined.
   \item Converging to class (4) means that the diffusion is erratic but
     homogeneous, in the sense that everywhere there are active and inactive
     agents, on odd and even positions. Any such configuration could happen,
     which means that there is no region in the network with specific
     properties.
  \end{itemize}

\section{Strict aggregation functions and nonbipartite networks}\label{sec:noblinker}

In this section, we focus on strict aggregation functions but we assume that the
graph is not bipartite. We saw in Sections
\ref{sec:decomp_absorbing_transient} and \ref{sec:decomp_ergodicity} that we
could separate what was happening at odd positions (respectively even positions)
at odd times from what was happening at odd positions (respectively even
positions) at even times. This is not possible anymore, so that the results of
Section \ref{sec:decomp_absorbing_transient} cannot be applied. On the other
hand, we will see that any graph admits at least one bipartite
  subgraph. Since a path in the subgraph is still a path in the graph, the
procedures of Section \ref{sec:decomp_ergodicity} can still be used.

\begin{example}
Consider $\scX=\ZZ^2$ and the $\sqrt{2}$-neighborhood (for the
Euclidean distance) defined by, for any $(n,m) \in \scX$,
\[
\Gamma((n,m))=\{(p,q), \ p\in\{n-1,n,n+1\} \text{ and } q\in \{m-1,m,m+1\}\}.
\]
(see Figure~\ref{fig:subgr} (left)). This graph is not bipartite. Nevertheless, by focusing on the $1$-neighborhood, we can still define
\[
\scX_e=\{(n,m) \text{ s.t. } n+m \text{ is even}\}
\]
and
\[
\scX_o=\{(n,m) \text{ s.t. } n+m \text{ is odd}\}.
\]
This amounts to considering the subgraph where all diagonal links (i.e., links
  internal to $\scX_e$ and $\scX_o$) have
  been deleted. This subgraph is bipartite (see Figure~\ref{fig:subgr}
  (middle)) and we can apply the procedures of Section
  \ref{sec:decomp_ergodicity}. 
\end{example}

%
%

More generally, any network admits a bipartite subnetwork. Indeed, given $(\scX,\sim)$ and $c\in \scX$, one can define the distance from any node $x$ to $c$:
\[
\forall x\in \scX, d(x,c)=\inf \{n \in \NN, \exists (x_l)_{1 \leq l\leq n} \text{ path between } x \text{ and } c\}.
\]
This yields naturally two sets 
\[
\scX^c_o=\{x\in \scX, \text{ s.t. } d(x,c) \text{ is odd}\}
\]
and 
\[
\scX^c_e=\{x\in \scX, \text{ s.t. } d(x,c) \text{ is even}\}.
\]
Let us define from $\sim$ the new relation $\approx^c$ on $\scX$ as follows: internal edges
in $\scX^c_o$ and in $\scX^c_e$ for $\sim$ are deleted, and only edges between a
node in $\scX^c_o$ and $\scX^c_e$ for $\sim$ remain. The induced subnetwork
  $(\scX,\approx^c)$ has a bipartition exactly given by $\scX^c_o$ and
$\scX^c_e$. Moreover, first this subgraph is connected, because by construction
  every node is connected to $c$, and second  every node has a finite number of
  neighbors. Hence, $\approx^c$ is a neighborhood relation.
This yields the following proposition.

\begin{proposition}
Given a network $(\scX,\sim)$, there exists a subnetwork $(\scX,\approx)$ 
  which is bipartite.
\end{proposition}

\begin{remark}
  \begin{enumerate}
  \item Suppose $(\scX,\sim)$ is not bipartite and consider a 
    bipartite subgraph
    $(\scX,\approx)$. Does there always exist a node $c$ such
    that $\approx$ is equal to $\approx^c$, the relation induced by the distance to
    $c$? In other words, is the above mechanism able to generate any
      bipartite subgraph? The answer is no, as shown by the above example with
    $\scX=\ZZ^2$. Figure~\ref{fig:subgr} (right) shows the subgraph induced by the
    distance to the central node $(0,0)$ (or to any arbitrary node), which fairly
    differs from the one on the middle figure.
  \item 
    Notice that there may exist several bipartite subgraphs, as shown
    in Figure ~\ref{fig:subgr}. The choice of the subgraph is important, since it is
    possible for some partial configuration to be stored for one subgraph and
    not for another one, as shown by Example~\ref{exa:choice_blinker}.
  \end{enumerate}
\end{remark}

  \begin{figure}[htb]
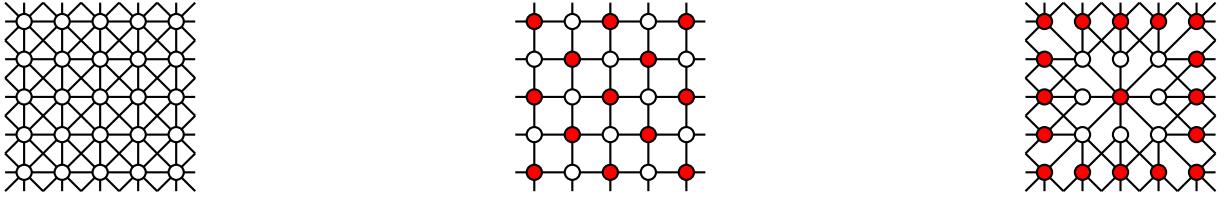

    \begin{center}
\psset{unit=0.5cm}
\pspicture(-0.5,-1.5)(4.5,4.5)
\psline(-0.5,0)(4.5,0)
\psline(-0.5,1)(4.5,1)
\psline(-0.5,2)(4.5,2)
\psline(-0.5,3)(4.5,3)
\psline(-0.5,4)(4.5,4)
\psline(0,-0.5)(0,4.5)
\psline(1,-0.5)(1,4.5)
\psline(2,-0.5)(2,4.5)
\psline(3,-0.5)(3,4.5)
\psline(4,-0.5)(4,4.5)
\psline(-0.5,-0.5)(4.5,4.5)
\psline(0.5,-0.5)(4.5,3.5)
\psline(1.5,-0.5)(4.5,2.5)
\psline(2.5,-0.5)(4.5,1.5)
\psline(3.5,-0.5)(4.5,0.5)
\psline(-0.5,0.5)(3.5,4.5)
\psline(-0.5,1.5)(2.5,4.5)
\psline(-0.5,2.5)(1.5,4.5)
\psline(-0.5,3.5)(0.5,4.5)
\psline(-0.5,4.5)(4.5,-0.5)
\psline(0.5,4.5)(4.5,0.5)
\psline(1.5,4.5)(4.5,1.5)
\psline(2.5,4.5)(4.5,2.5)
\psline(3.5,4.5)(4.5,3.5)
\psline(-0.5,3.5)(3.5,-0.5)
\psline(-0.5,2.5)(2.5,-0.5)
\psline(-0.5,1.5)(1.5,-0.5)
\psline(-0.5,0.5)(0.5,-0.5)
\pscircle[fillstyle=solid](0,0){0.2}
\pscircle[fillstyle=solid](0,1){0.2}
\pscircle[fillstyle=solid](0,2){0.2}
\pscircle[fillstyle=solid](0,3){0.2}
\pscircle[fillstyle=solid](0,4){0.2}
\pscircle[fillstyle=solid](1,0){0.2}
\pscircle[fillstyle=solid](1,1){0.2}
\pscircle[fillstyle=solid](1,2){0.2}
\pscircle[fillstyle=solid](1,3){0.2}
\pscircle[fillstyle=solid](1,4){0.2}
\pscircle[fillstyle=solid](2,0){0.2}
\pscircle[fillstyle=solid](2,1){0.2}
\pscircle[fillstyle=solid](2,2){0.2}
\pscircle[fillstyle=solid](2,3){0.2}
\pscircle[fillstyle=solid](2,4){0.2}
\pscircle[fillstyle=solid](3,0){0.2}
\pscircle[fillstyle=solid](3,1){0.2}
\pscircle[fillstyle=solid](3,2){0.2}
\pscircle[fillstyle=solid](3,3){0.2}
\pscircle[fillstyle=solid](3,4){0.2}
\pscircle[fillstyle=solid](4,0){0.2}
\pscircle[fillstyle=solid](4,1){0.2}
\pscircle[fillstyle=solid](4,2){0.2}
\pscircle[fillstyle=solid](4,3){0.2}
\pscircle[fillstyle=solid](4,4){0.2}
\endpspicture
\hfill
\pspicture(-0.5,-1.5)(4.5,4.5)
\psline(-0.5,0)(4.5,0)
\psline(-0.5,1)(4.5,1)
\psline(-0.5,2)(4.5,2)
\psline(-0.5,3)(4.5,3)
\psline(-0.5,4)(4.5,4)
\psline(0,-0.5)(0,4.5)
\psline(1,-0.5)(1,4.5)
\psline(2,-0.5)(2,4.5)
\psline(3,-0.5)(3,4.5)
\psline(4,-0.5)(4,4.5)
\pscircle[fillstyle=solid,fillcolor=red](0,0){0.2}
\pscircle[fillstyle=solid](0,1){0.2}
\pscircle[fillstyle=solid,fillcolor=red](0,2){0.2}
\pscircle[fillstyle=solid](0,3){0.2}
\pscircle[fillstyle=solid,fillcolor=red](0,4){0.2}
\pscircle[fillstyle=solid](1,0){0.2}
\pscircle[fillstyle=solid,fillcolor=red](1,1){0.2}
\pscircle[fillstyle=solid](1,2){0.2}
\pscircle[fillstyle=solid,fillcolor=red](1,3){0.2}
\pscircle[fillstyle=solid](1,4){0.2}
\pscircle[fillstyle=solid,fillcolor=red](2,0){0.2}
\pscircle[fillstyle=solid](2,1){0.2}
\pscircle[fillstyle=solid,fillcolor=red](2,2){0.2}
\pscircle[fillstyle=solid](2,3){0.2}
\pscircle[fillstyle=solid,fillcolor=red](2,4){0.2}
\pscircle[fillstyle=solid](3,0){0.2}
\pscircle[fillstyle=solid,fillcolor=red](3,1){0.2}
\pscircle[fillstyle=solid](3,2){0.2}
\pscircle[fillstyle=solid,fillcolor=red](3,3){0.2}
\pscircle[fillstyle=solid](3,4){0.2}
\pscircle[fillstyle=solid,fillcolor=red](4,0){0.2}
\pscircle[fillstyle=solid](4,1){0.2}
\pscircle[fillstyle=solid,fillcolor=red](4,2){0.2}
\pscircle[fillstyle=solid](4,3){0.2}
\pscircle[fillstyle=solid,fillcolor=red](4,4){0.2}
\endpspicture
\hfill
\pspicture(-0.5,-1.5)(4.5,4.5)
\psline(-0.5,0)(0,0)\psline(4,0)(4.5,0)
\psline(-0.5,4)(0,4)\psline(4,4)(4.5,4)
\psline(0,-0.5)(0,0)\psline(0,4)(0,4.5)
\psline(4,-0.5)(4,0)\psline(4,4)(4,4.5)
\psline(-0.5,1)(1,1)\psline(3,1)(4.5,1)
\psline(-0.5,2)(4.5,2)
\psline(-0.5,3)(1,3)\psline(3,3)(4.5,3)
\psline(1,-0.5)(1,1)\psline(1,3)(1,4.5)
\psline(2,-0.5)(2,4.5)
\psline(3,-0.5)(3,1)\psline(3,3)(3,4.5)
\psline(-0.5,-0.5)(4.5,4.5)
\psline(0.5,-0.5)(2,1)\psline(3,2)(4.5,3.5)
\psline(1.5,-0.5)(4.5,2.5)
\psline(2.5,-0.5)(3,0)\psline(4,1)(4.5,1.5)
\psline(3.5,-0.5)(4.5,0.5)
\psline(-0.5,0.5)(1,2)\psline(2,3)(3.5,4.5)
\psline(-0.5,1.5)(2.5,4.5)
\psline(-0.5,2.5)(0,3)\psline(1,4)(1.5,4.5)
\psline(-0.5,3.5)(0.5,4.5)
\psline(-0.5,4.5)(4.5,-0.5)
\psline(0.5,4.5)(2,3)\psline(3,2)(4.5,0.5)
\psline(1.5,4.5)(4.5,1.5)
\psline(2.5,4.5)(3,4)\psline(4,3)(4.5,2.5)
\psline(3.5,4.5)(4.5,3.5)
\psline(-0.5,3.5)(1,2)\psline(2,1)(3.5,-0.5)
\psline(-0.5,2.5)(2.5,-0.5)
\psline(-0.5,1.5)(0,1)\psline(1,0)(1.5,-0.5)
\psline(-0.5,0.5)(0.5,-0.5)
\pscircle[fillstyle=solid,fillcolor=red](0,0){0.2}
\pscircle[fillstyle=solid,fillcolor=red](0,1){0.2}
\pscircle[fillstyle=solid,fillcolor=red](0,2){0.2}
\pscircle[fillstyle=solid,fillcolor=red](0,3){0.2}
\pscircle[fillstyle=solid,fillcolor=red](0,4){0.2}
\pscircle[fillstyle=solid,fillcolor=red](1,0){0.2}
\pscircle[fillstyle=solid](1,1){0.2}
\pscircle[fillstyle=solid](1,2){0.2}
\pscircle[fillstyle=solid](1,3){0.2}
\pscircle[fillstyle=solid,fillcolor=red](1,4){0.2}
\pscircle[fillstyle=solid,fillcolor=red](2,0){0.2}
\pscircle[fillstyle=solid](2,1){0.2}
\pscircle[fillstyle=solid,fillcolor=red](2,2){0.2}
\pscircle[fillstyle=solid](2,3){0.2}
\pscircle[fillstyle=solid,fillcolor=red](2,4){0.2}
\pscircle[fillstyle=solid,fillcolor=red](3,0){0.2}
\pscircle[fillstyle=solid](3,1){0.2}
\pscircle[fillstyle=solid](3,2){0.2}
\pscircle[fillstyle=solid](3,3){0.2}
\pscircle[fillstyle=solid,fillcolor=red](3,4){0.2}
\pscircle[fillstyle=solid,fillcolor=red](4,0){0.2}
\pscircle[fillstyle=solid,fillcolor=red](4,1){0.2}
\pscircle[fillstyle=solid,fillcolor=red](4,2){0.2}
\pscircle[fillstyle=solid,fillcolor=red](4,3){0.2}
\pscircle[fillstyle=solid,fillcolor=red](4,4){0.2}
\endpspicture
    \end{center}
    \caption{The $\ZZ^2$ network with the $\sqrt{2}$-neighborhood (left); The
      $\ZZ^2$ subnetwork with the 1-neighborhood (middle); The $(\ZZ^2,\approx^c)$ subnetwork
      where $c=(0,0)$ (right). White and red nodes indicate the bipartition.}
    \label{fig:subgr}
  \end{figure}

\begin{example}\label{exa:choice_blinker}
We slightly modify the network $(\scX,\sim)$ of Example~\ref{counter_collision} as
  follows:
\begin{itemize}
\item $\Gamma^\sim(\alpha)=\{(0,0),\beta\}$,
\item $\Gamma^\sim(\beta)=\{(0,0),\alpha\}$,
\end{itemize}
i.e., $\alpha$ and $\beta$ are linked together. This network is not bipartite since we have a cycle of length $3$. If we consider the graph generated by
the distance to $(0,0)$ as subgraph, we obtain exactly the network of Example
\ref{counter_collision} where some configurations cannot be stored. On the
contrary, let us consider the subgraph $(\scX,\approx)$ where the link between $(0,0)$
and $\beta$ has been deleted, hence:
\begin{itemize}
\item $\Gamma^\approx(\alpha)=\{(0,0),\beta\}$,
\item $\Gamma^\approx(\beta)=\{\alpha\}$,
\item $\Gamma^\approx((0,0))=\{(-1,0),(0,-1),(1,0),(0,1),\alpha\}$,
\item for every $(n,m)\neq (0,0)$, $\Gamma^\approx((n,m))=\Gamma^\sim((n,m))$.
\end{itemize}
Then, this new graph is bipartite and every configuration can be stored by storing separately what happens on $\{\alpha,\beta\}$ and $\ZZ^2$.
\end{example}

%

%
Let $(\scX,\approx)$ be a bipartite subnetwork of $(\scX,\sim)$ and
let $(\scX_e,\scX_o)$ be the bipartition. We consider the
partition of the set of configurations where we count how many $1$ and $0$ exist
in $\omega$, without separating odd and even positions. Formally, a block of the
partition is a $2$-uple $(a,b)$ such that $a,b \in \{0,\F,\infty\}$ and there are
$a$ positions with value $0$ and $b$ positions with value $1$. Since there is an
infinite number of positions, there are only $5$ different blocks of the
partition which are nonempty.

\begin{theorem}\label{dec}
Assume that there exists $(\scX,\approx)$ a bipartite subnetwork of $(\scX,\sim)$ that satisfies the Richness Assumption. We have the following decomposition:
\begin{itemize}
\item[-] $(\infty,0)$ and $(0,\infty)$ are finite $\phi$-irreducible classes,
\item[-] $(\infty,\infty)$ is an infinite $\phi$-irreducible class,
\item[-] $(\infty,\F)$ and $(\F,\infty)$ are transient and $\phi$-irreducible sets.
\end{itemize}
\end{theorem}

(see proof in Appendix \ref{appD})

Compared to the results where the network is bipartite
(Corollary~\ref{cor:blk}), we obtain results of the same nature but much
simpler. The disappearance of the distinction between odd and even positions has
wippen out all cycles and periodic classes, since the periodicity was only due
to the alternation of a pattern on odd and even positions. In the two finite
classes we recognize the classes (1) and (2), the infinite class is the former
class (4), while all previous transient $\phi$-irreducible sets reduce only to
two such sets. Our conclusion and interpretation given for the case of
  bipartite networks therefore remains identical, up to the disappearance of
the periodic classes.

\section{Boolean aggregation functions}\label{sec:boolean}
The aggregation function being Boolean, the diffusion model (\ref{eq:1}) becomes
deterministic. In the previous sections, we used heavily the fact that the aggregation function was strict. Since it is not true anymore, we obtain very different results. Let us restrict our attention to anonymous functions. Then there
exists a threshold $0\leq\widetilde{q}\leq\gamma$ such that
\begin{equation}\label{eq:5a}
A(1_S) = \begin{cases}
  1 & \text{if } |S|\geq \widetilde{q}\\
  0 & \text{otherwise.}
  \end{cases}
\end{equation}
This yields the contagion model of \citet{mor00} which is also the
  threshold model introduced, e.g., by \citet{gra78}. In this model, the rule
of contagion with threshold $0\leq q\leq 1$ is the following. Given a
configuration $X(t)$ at time $t$, next configuration $X(t+1)$ is the set of
agents having a proportion of neighbors in $X(t)$ at least equal to $q$:
\begin{equation}\label{eq:6}
X(t+1) = \left\{x\in\scX\mymid \frac{|\Gamma(x)\cap X(t)|}{|\Gamma(x)|}\geq q\right\}.
\end{equation}
Comparing (\ref{eq:6}) with (\ref{eq:1}) and (\ref{eq:5a}) yields the equivalence
between the two models. The \emph{contagion threshold} $\xi$ is the largest $q$
such that infection spreads over $\scX$ from some finite group
$X(0)$.  \citet{mor00} has shown that for any network (in the sense of Section~\ref{subsec:societies}), $\xi\leq\frac{1}{2}$.

By definition, if $q$ is below the contagion threshold, then the
absorbing states are the trivial states $\varnothing$ and $\scX$. Otherwise, other
nontrivial absorbing classes may occur, as will be shown.

\medskip

Let us try to find first absorbing states. As we noted earlier, $X$ is an absorbing
state if it is a fixed point, i.e., $X(t+1)=X(t)$. We introduce the
  following convenient notation. The {\it frontier} of $X$ is the set
  \[
\partial X= \cloX\setminus\intX = \{x\in \scX\mymid \Gamma(x)\cap
X^c\neq\varnothing \text{ and } \Gamma(x)\cap X\neq\varnothing\}.
  \]
We also introduce $\partial^+X = \partial X\cap X$ the {\it inner frontier} of $X$,
and $\partial^-X = \partial X\cap X^c$ the {\it outer frontier} of $X$. The following results are
immediate.
\begin{proposition}
Suppose $A$ is an anonymous Boolean aggregation function. $X$ is an absorbing state
if and only if:
\begin{enumerate}
\item Each inner frontier node $x\in\partial^+X$ has at least $\lceil \gamma q\rceil$
  neighbors in $X$, and
\item Each outer frontier node $x\in\partial^-X$ has at least $\lfloor
  \gamma(1-q)\rfloor+1$ neighbors in $X^c$.
\end{enumerate}
\end{proposition}
\begin{proof}
$X$ will not lose an element if each element of $X$ has at least $\gamma q $
neighbors in $X$, and $X$ will not attract a new element if each element of
$\scX\setminus X$ has strictly less than $\gamma q$ neighbors in $X$, hence  at
least $\lfloor \gamma(1-q)\rfloor+1$ neighbors in $\scX\setminus X$. 
\end{proof}

We examine two particular networks from Example \ref{ex:Z2} for illustration.

\paragraph{The case $\ZZ^2$ with the 1-neighborhood}

Let us consider various values of
$q$. The following table gives the minimal number of neighbors in $X$ for frontier nodes, so that $X$ is an absorbing state. 

\begin{center}
\begin{tabular}{|p{5cm}|cccc|}\hline
$q$ & 1/4 & 1/2 & 3/4  & 1\\ \hline
{\footnotesize minimal nb of interior neighbors for nodes in $X$} & 1 & 2 & 3 & 4
\\ \hline
{\footnotesize minimal nb of interior neighbors for nodes outside $X$} & 4 & 3 & 2 & 1
\\ \hline
\end{tabular}
\end{center} 
Let us describe the shape of the frontier of a connected set $X$ with $|X|>1$
satisfying the above requirements. To this end, we consider a frontier node $x$
of $X$ and the number of its interior neighbors $\eta (x) := |\Gamma(x)\cap X|$.  
\begin{enumerate}
\item $\eta (x)\geq 1$ for every frontier node $x$. This is true for any $X$ such that $|X|\geq 2$.
\item $\eta (x)\geq 2$ for every frontier node $x$. Only the situation
  $\eta(x)=1$ is forbidden, which corresponds to ``antennas'' (see
  Figure~\ref{fig:2}, left). Therefore, $X$ is any shape without antennas, with
  $|X|\geq 3$.
\item $\eta(x)\geq 3$ for every frontier node $x$. Then additional forbidden
  situations are ``convex corners'' and ``isthms'' (see
  Figure~\ref{fig:2}). Then
  $X$ is any shape without antennas, convex corners and isthms, and $|X|\geq
  4$. Note that since $X$ may be infinite, one can have shapes with concave
  corners (which are allowed) but without convex corners (see Figure~\ref{fig:3}).  
\item Evidently, $\eta(x)\geq 4$ is impossible by definition of frontier nodes.
\end{enumerate}
\begin{figure}[htb]
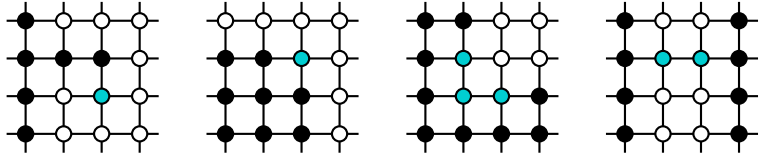
 
\begin{center}
\psset{unit=0.5cm}
\pspicture(-0.5,-0.5)(3.5,3.5)
\psline(-0.5,0)(3.5,0)
\psline(-0.5,1)(3.5,1)
\psline(-0.5,2)(3.5,2)
\psline(-0.5,3)(3.5,3)
\psline(0,-0.5)(0,3.5)
\psline(1,-0.5)(1,3.5)
\psline(2,-0.5)(2,3.5)
\psline(3,-0.5)(3,3.5)
\pscircle[fillstyle=solid,fillcolor=black](0,0){0.2}
\pscircle[fillstyle=solid,fillcolor=black](0,1){0.2}
\pscircle[fillstyle=solid,fillcolor=black](0,2){0.2}
\pscircle[fillstyle=solid,fillcolor=black](0,3){0.2}
\pscircle[fillstyle=solid](1,0){0.2}
\pscircle[fillstyle=solid](1,1){0.2}
\pscircle[fillstyle=solid,fillcolor=black](1,2){0.2}
\pscircle[fillstyle=solid](1,3){0.2}
\pscircle[fillstyle=solid](2,0){0.2}
\pscircle[fillstyle=solid,fillcolor=cyan](2,1){0.2}
\pscircle[fillstyle=solid,fillcolor=black](2,2){0.2}
\pscircle[fillstyle=solid](2,3){0.2}
\pscircle[fillstyle=solid](3,0){0.2}
\pscircle[fillstyle=solid](3,1){0.2}
\pscircle[fillstyle=solid](3,2){0.2}
\pscircle[fillstyle=solid](3,3){0.2}
\endpspicture
\hspace*{0.5cm}
\pspicture(-0.5,-0.5)(3.5,3.5)
\psline(-0.5,0)(3.5,0)
\psline(-0.5,1)(3.5,1)
\psline(-0.5,2)(3.5,2)
\psline(-0.5,3)(3.5,3)
\psline(0,-0.5)(0,3.5)
\psline(1,-0.5)(1,3.5)
\psline(2,-0.5)(2,3.5)
\psline(3,-0.5)(3,3.5)
\pscircle[fillstyle=solid,fillcolor=black](0,0){0.2}
\pscircle[fillstyle=solid,fillcolor=black](0,1){0.2}
\pscircle[fillstyle=solid,fillcolor=black](0,2){0.2}
\pscircle[fillstyle=solid](0,3){0.2}
\pscircle[fillstyle=solid,fillcolor=black](1,0){0.2}
\pscircle[fillstyle=solid,fillcolor=black](1,1){0.2}
\pscircle[fillstyle=solid,fillcolor=black](1,2){0.2}
\pscircle[fillstyle=solid](1,3){0.2}
\pscircle[fillstyle=solid,fillcolor=black](2,0){0.2}
\pscircle[fillstyle=solid,fillcolor=black](2,1){0.2}
\pscircle[fillstyle=solid,fillcolor=cyan](2,2){0.2}
\pscircle[fillstyle=solid](2,3){0.2}
\pscircle[fillstyle=solid](3,0){0.2}
\pscircle[fillstyle=solid](3,1){0.2}
\pscircle[fillstyle=solid](3,2){0.2}
\pscircle[fillstyle=solid](3,3){0.2}
\endpspicture
\hspace*{0.5cm}
\pspicture(-0.5,-0.5)(3.5,3.5)
\psline(-0.5,0)(3.5,0)
\psline(-0.5,1)(3.5,1)
\psline(-0.5,2)(3.5,2)
\psline(-0.5,3)(3.5,3)
\psline(0,-0.5)(0,3.5)
\psline(1,-0.5)(1,3.5)
\psline(2,-0.5)(2,3.5)
\psline(3,-0.5)(3,3.5)
\pscircle[fillstyle=solid,fillcolor=black](0,0){0.2}
\pscircle[fillstyle=solid,fillcolor=black](0,1){0.2}
\pscircle[fillstyle=solid,fillcolor=black](0,2){0.2}
\pscircle[fillstyle=solid,fillcolor=black](0,3){0.2}
\pscircle[fillstyle=solid,fillcolor=black](1,0){0.2}
\pscircle[fillstyle=solid,fillcolor=cyan](1,1){0.2}
\pscircle[fillstyle=solid,fillcolor=cyan](1,2){0.2}
\pscircle[fillstyle=solid,fillcolor=black](1,3){0.2}
\pscircle[fillstyle=solid,fillcolor=black](2,0){0.2}
\pscircle[fillstyle=solid,fillcolor=cyan](2,1){0.2}
\pscircle[fillstyle=solid](2,2){0.2}
\pscircle[fillstyle=solid](2,3){0.2}
\pscircle[fillstyle=solid,fillcolor=black](3,0){0.2}
\pscircle[fillstyle=solid,fillcolor=black](3,1){0.2}
\pscircle[fillstyle=solid](3,2){0.2}
\pscircle[fillstyle=solid](3,3){0.2}
\endpspicture
\hspace*{0.5cm}
\pspicture(-0.5,-0.5)(3.5,3.5)
\psline(-0.5,0)(3.5,0)
\psline(-0.5,1)(3.5,1)
\psline(-0.5,2)(3.5,2)
\psline(-0.5,3)(3.5,3)
\psline(0,-0.5)(0,3.5)
\psline(1,-0.5)(1,3.5)
\psline(2,-0.5)(2,3.5)
\psline(3,-0.5)(3,3.5)
\pscircle[fillstyle=solid,fillcolor=black](0,0){0.2}
\pscircle[fillstyle=solid,fillcolor=black](0,1){0.2}
\pscircle[fillstyle=solid,fillcolor=black](0,2){0.2}
\pscircle[fillstyle=solid,fillcolor=black](0,3){0.2}
\pscircle[fillstyle=solid](1,0){0.2}
\pscircle[fillstyle=solid](1,1){0.2}
\pscircle[fillstyle=solid,fillcolor=cyan](1,2){0.2}
\pscircle[fillstyle=solid](1,3){0.2}
\pscircle[fillstyle=solid](2,0){0.2}
\pscircle[fillstyle=solid](2,1){0.2}
\pscircle[fillstyle=solid,fillcolor=cyan](2,2){0.2}
\pscircle[fillstyle=solid](2,3){0.2}
\pscircle[fillstyle=solid,fillcolor=black](3,0){0.2}
\pscircle[fillstyle=solid,fillcolor=black](3,1){0.2}
\pscircle[fillstyle=solid,fillcolor=black](3,2){0.2}
\pscircle[fillstyle=solid,fillcolor=black](3,3){0.2}
\endpspicture
\end{center}
\caption{4$\times$4 portion of $\ZZ^2$, where $X$ is in black or
  blue. By convention, the color of nodes extends infinitely in any direction. From left to right: antenna ($x$ such that $\eta(x)=1$ in blue), convex
  corner ($x$ such that $\eta(x)=2$ in blue), concave corner in blue, isthm ($x$ such that $\eta(x)=2$ in blue)}
\label{fig:2}
\end{figure}
From the previous analysis, it follows that only $q=\frac{1}{2}$ or
$\frac{3}{4}$ may lead to nontrivial absorbing classes. Also, these two cases are
exact complements of each other, in the sense that $X$ is a possible absorbing
state for $q=\frac{1}{2}$ if and only if $\scX\setminus X$ is a possible
absorbing state for $q=\frac{3}{4}$. Taking for example the latter, each
connected component of $X$ should be of size at least 4, and should have no
convex corner, no antenna and no isthm, while each connected component of the
complement set should be of size at least 3 and have no antennas. Obvious
examples are half-planes, concave corners extending infinitely in both
directions, and infinite shapes with squared ``holes'' of width greater than 1
(see the example in Figure~\ref{fig:3}). 
\begin{figure}[htp]
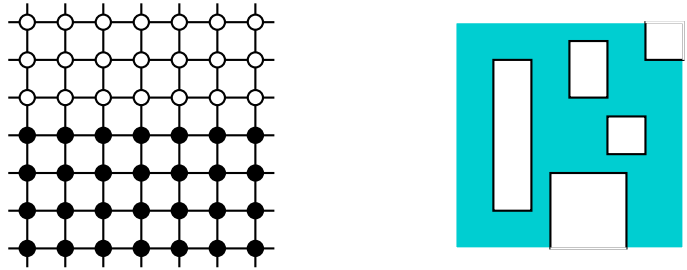

\begin{center}
\psset{unit=0.5cm}
\pspicture(-0.5,-0.5)(6.5,6.5)
\psline(-0.5,0)(6.5,0)
\psline(-0.5,1)(6.5,1)
\psline(-0.5,2)(6.5,2)
\psline(-0.5,3)(6.5,3)
\psline(-0.5,4)(6.5,4)
\psline(-0.5,5)(6.5,5)
\psline(-0.5,6)(6.5,6)
\psline(0,-0.5)(0,6.5)
\psline(1,-0.5)(1,6.5)
\psline(2,-0.5)(2,6.5)
\psline(3,-0.5)(3,6.5)
\psline(4,-0.5)(4,6.5)
\psline(5,-0.5)(5,6.5)
\psline(6,-0.5)(6,6.5)
\pscircle[fillstyle=solid,fillcolor=black](0,0){0.2}
\pscircle[fillstyle=solid,fillcolor=black](1,0){0.2}
\pscircle[fillstyle=solid,fillcolor=black](2,0){0.2}
\pscircle[fillstyle=solid,fillcolor=black](3,0){0.2}
\pscircle[fillstyle=solid,fillcolor=black](4,0){0.2}
\pscircle[fillstyle=solid,fillcolor=black](5,0){0.2}
\pscircle[fillstyle=solid,fillcolor=black](6,0){0.2}
\pscircle[fillstyle=solid,fillcolor=black](0,1){0.2}
\pscircle[fillstyle=solid,fillcolor=black](1,1){0.2}
\pscircle[fillstyle=solid,fillcolor=black](2,1){0.2}
\pscircle[fillstyle=solid,fillcolor=black](3,1){0.2}
\pscircle[fillstyle=solid,fillcolor=black](4,1){0.2}
\pscircle[fillstyle=solid,fillcolor=black](5,1){0.2}
\pscircle[fillstyle=solid,fillcolor=black](6,1){0.2}
\pscircle[fillstyle=solid,fillcolor=black](0,2){0.2}
\pscircle[fillstyle=solid,fillcolor=black](1,2){0.2}
\pscircle[fillstyle=solid,fillcolor=black](2,2){0.2}
\pscircle[fillstyle=solid,fillcolor=black](3,2){0.2}
\pscircle[fillstyle=solid,fillcolor=black](4,2){0.2}
\pscircle[fillstyle=solid,fillcolor=black](5,2){0.2}
\pscircle[fillstyle=solid,fillcolor=black](6,2){0.2}
\pscircle[fillstyle=solid,fillcolor=black](0,3){0.2}
\pscircle[fillstyle=solid,fillcolor=black](1,3){0.2}
\pscircle[fillstyle=solid,fillcolor=black](2,3){0.2}
\pscircle[fillstyle=solid,fillcolor=black](3,3){0.2}
\pscircle[fillstyle=solid,fillcolor=black](4,3){0.2}
\pscircle[fillstyle=solid,fillcolor=black](5,3){0.2}
\pscircle[fillstyle=solid,fillcolor=black](6,3){0.2}
\pscircle[fillstyle=solid](0,4){0.2}
\pscircle[fillstyle=solid](1,4){0.2}
\pscircle[fillstyle=solid](2,4){0.2}
\pscircle[fillstyle=solid](3,4){0.2}
\pscircle[fillstyle=solid](4,4){0.2}
\pscircle[fillstyle=solid](5,4){0.2}
\pscircle[fillstyle=solid](6,4){0.2}
\pscircle[fillstyle=solid](0,5){0.2}
\pscircle[fillstyle=solid](1,5){0.2}
\pscircle[fillstyle=solid](2,5){0.2}
\pscircle[fillstyle=solid](3,5){0.2}
\pscircle[fillstyle=solid](4,5){0.2}
\pscircle[fillstyle=solid](5,5){0.2}
\pscircle[fillstyle=solid](6,5){0.2}
\pscircle[fillstyle=solid](0,6){0.2}
\pscircle[fillstyle=solid](1,6){0.2}
\pscircle[fillstyle=solid](2,6){0.2}
\pscircle[fillstyle=solid](3,6){0.2}
\pscircle[fillstyle=solid](4,6){0.2}
\pscircle[fillstyle=solid](5,6){0.2}
\pscircle[fillstyle=solid](6,6){0.2}
\endpspicture
\hspace*{2cm}
\psset{unit=0.5cm}
\pspicture(-0.5,-0.5)(6.5,6.5)
\pspolygon[fillstyle=solid,fillcolor=cyan,linecolor=white](0,0)(0,6)(6,6)(6,0)
\pspolygon[fillstyle=solid](1,1)(1,5)(2,5)(2,1)
\pspolygon[fillstyle=solid](3,4)(3,5.5)(4,5.5)(4,4)
\pspolygon[fillstyle=solid](2.5,0)(2.5,2)(4.5,2)(4.5,0)
\psline[linecolor=white](2.5,0)(4.5,0)
\pspolygon[fillstyle=solid](5,5)(5,6)(6,6)(6,5)
\psline[linecolor=white](5,6)(6,6)(6,5)
\pspolygon[fillstyle=solid](4,2.5)(4,3.5)(5,3.5)(5,2.5)
\endpspicture
\end{center}
\caption{Examples of absorbing states with $q=\frac{3}{4}$. Left: horizontal
  half-plane. Right: shape with squared holes (in blue; the grid is not represented)}
\label{fig:3}
\end{figure}

\paragraph{The case of $\ZZ^2$ with the $\sqrt{2}$-neighborhood}

We obtain the following table.
\begin{center}
\begin{tabular}{|p{5cm}|ccccccc|}\hline
$q$ & 1/8 & 1/4 & 3/8 & 1/2 & 5/8 & 3/4 & 7/8 \\ \hline
{\footnotesize minimal nb of interior neighbors for nodes in $X$} & 1 & 2 & 3
& 4 & 5 & 6 & 7 \\ \hline
{\footnotesize minimal nb of interior neighbors for nodes outside $X$} & 8 & 7 & 6
& 5 & 4 & 3 & 2 \\ \hline
\end{tabular}
\end{center} 
The analysis of possible shapes in this topology reveals to be much more
complex and we will not provide a full analysis as in the previous
example. Obviously, $q=\frac{1}{8}$ is ruled out. If $\eta (x)=7$, then the only
possibility is that $\scX\setminus X$ is a singleton. Since then $\scX\setminus
X$ cannot satisfy the condition $\eta (x)\geq 2$, this rules out $q=\frac{1}{4}$
and $q=\frac{7}{8}$. If $\eta (x)\geq 3$, as before, antennas and isthms of
length greater than 2 are forbidden, but not convex corners. For $\eta (x)\geq 4$,
squared convex corners are no more allowed (but $45^\circ$ corners are
possible). For $\eta (x)\geq 5$, no corners are allowed except concave ones. For
$\eta (x)\geq 6$, a remarkable situation is where $\scX\setminus X$ is a square
with edges of length 1. Therefore, as possible absorbing states (we do not
pretend to be exhaustive) we find:
\begin{enumerate}
\item $q=3/8$: a connected component of $X$ is a square with edge of length 1 (see Figure~\ref{fig:4}, left)
\item $q=1/2$ or $q=5/8$: the two cases are complements of each other. Taking the
  latter, a connected component of $X$ is any shape of size at least 6 without
  convex corners, antennas nor isthms, and each connected component of the
  complement set should be of size at least 5 and have no antennas, isthms and
  squared convex corners. For example, half-planes are allowed, as well as
  infinite shapes with squared holes whose corners are ``rounded'' (see
  Figure~\ref{fig:4}, right).
\item $q=3/4$: a connected component of $\scX\setminus S$ is a square with edge of length 1. 
\end{enumerate}  
\begin{figure}[htp]
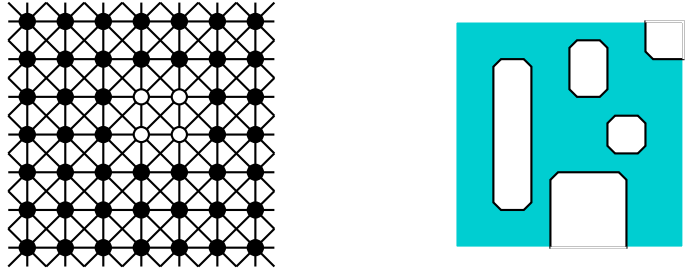

\begin{center}
\psset{unit=0.5cm}
\pspicture(-0.5,-0.5)(6.5,6.5)
\psline(-0.5,0)(6.5,0)
\psline(-0.5,1)(6.5,1)
\psline(-0.5,2)(6.5,2)
\psline(-0.5,3)(6.5,3)
\psline(-0.5,4)(6.5,4)
\psline(-0.5,5)(6.5,5)
\psline(-0.5,6)(6.5,6)
\psline(0,-0.5)(0,6.5)
\psline(1,-0.5)(1,6.5)
\psline(2,-0.5)(2,6.5)
\psline(3,-0.5)(3,6.5)
\psline(4,-0.5)(4,6.5)
\psline(5,-0.5)(5,6.5)
\psline(6,-0.5)(6,6.5)
\psline(-0.5,0.5)(0.5,-0.5)
\psline(-0.5,1.5)(1.5,-0.5)
\psline(-0.5,2.5)(2.5,-0.5)
\psline(-0.5,3.5)(3.5,-0.5)
\psline(-0.5,4.5)(4.5,-0.5)
\psline(-0.5,5.5)(5.5,-0.5)
\psline(-0.5,6.5)(6.5,-0.5)
\psline(0.5,6.5)(6.5,0.5)
\psline(1.5,6.5)(6.5,1.5)
\psline(2.5,6.5)(6.5,2.5)
\psline(3.5,6.5)(6.5,3.5)
\psline(4.5,6.5)(6.5,4.5)
\psline(5.5,6.5)(6.5,5.5)
\psline(-0.5,5.5)(0.5,6.5)
\psline(-0.5,4.5)(1.5,6.5)
\psline(-0.5,3.5)(2.5,6.5)
\psline(-0.5,2.5)(3.5,6.5)
\psline(-0.5,1.5)(4.5,6.5)
\psline(-0.5,0.5)(5.5,6.5)
\psline(-0.5,-0.5)(6.5,6.5)
\psline(0.5,-0.5)(6.5,5.5)
\psline(1.5,-0.5)(6.5,4.5)
\psline(2.5,-0.5)(6.5,3.5)
\psline(3.5,-0.5)(6.5,2.5)
\psline(4.5,-0.5)(6.5,1.5)
\psline(5.5,-0.5)(6.5,0.5)
\pscircle[fillstyle=solid,fillcolor=black](0,0){0.2}
\pscircle[fillstyle=solid,fillcolor=black](1,0){0.2}
\pscircle[fillstyle=solid,fillcolor=black](2,0){0.2}
\pscircle[fillstyle=solid,fillcolor=black](3,0){0.2}
\pscircle[fillstyle=solid,fillcolor=black](4,0){0.2}
\pscircle[fillstyle=solid,fillcolor=black](5,0){0.2}
\pscircle[fillstyle=solid,fillcolor=black](6,0){0.2}
\pscircle[fillstyle=solid,fillcolor=black](0,1){0.2}
\pscircle[fillstyle=solid,fillcolor=black](1,1){0.2}
\pscircle[fillstyle=solid,fillcolor=black](2,1){0.2}
\pscircle[fillstyle=solid,fillcolor=black](3,1){0.2}
\pscircle[fillstyle=solid,fillcolor=black](4,1){0.2}
\pscircle[fillstyle=solid,fillcolor=black](5,1){0.2}
\pscircle[fillstyle=solid,fillcolor=black](6,1){0.2}
\pscircle[fillstyle=solid,fillcolor=black](0,2){0.2}
\pscircle[fillstyle=solid,fillcolor=black](1,2){0.2}
\pscircle[fillstyle=solid,fillcolor=black](2,2){0.2}
\pscircle[fillstyle=solid,fillcolor=black](3,2){0.2}
\pscircle[fillstyle=solid,fillcolor=black](4,2){0.2}
\pscircle[fillstyle=solid,fillcolor=black](5,2){0.2}
\pscircle[fillstyle=solid,fillcolor=black](6,2){0.2}
\pscircle[fillstyle=solid,fillcolor=black](0,3){0.2}
\pscircle[fillstyle=solid,fillcolor=black](1,3){0.2}
\pscircle[fillstyle=solid,fillcolor=black](2,3){0.2}
\pscircle[fillstyle=solid](3,3){0.2}
\pscircle[fillstyle=solid](4,3){0.2}
\pscircle[fillstyle=solid,fillcolor=black](5,3){0.2}
\pscircle[fillstyle=solid,fillcolor=black](6,3){0.2}
\pscircle[fillstyle=solid,fillcolor=black](0,4){0.2}
\pscircle[fillstyle=solid,fillcolor=black](1,4){0.2}
\pscircle[fillstyle=solid,fillcolor=black](2,4){0.2}
\pscircle[fillstyle=solid](3,4){0.2}
\pscircle[fillstyle=solid](4,4){0.2}
\pscircle[fillstyle=solid,fillcolor=black](5,4){0.2}
\pscircle[fillstyle=solid,fillcolor=black](6,4){0.2}
\pscircle[fillstyle=solid,fillcolor=black](0,5){0.2}
\pscircle[fillstyle=solid,fillcolor=black](1,5){0.2}
\pscircle[fillstyle=solid,fillcolor=black](2,5){0.2}
\pscircle[fillstyle=solid,fillcolor=black](3,5){0.2}
\pscircle[fillstyle=solid,fillcolor=black](4,5){0.2}
\pscircle[fillstyle=solid,fillcolor=black](5,5){0.2}
\pscircle[fillstyle=solid,fillcolor=black](6,5){0.2}
\pscircle[fillstyle=solid,fillcolor=black](0,6){0.2}
\pscircle[fillstyle=solid,fillcolor=black](1,6){0.2}
\pscircle[fillstyle=solid,fillcolor=black](2,6){0.2}
\pscircle[fillstyle=solid,fillcolor=black](3,6){0.2}
\pscircle[fillstyle=solid,fillcolor=black](4,6){0.2}
\pscircle[fillstyle=solid,fillcolor=black](5,6){0.2}
\pscircle[fillstyle=solid,fillcolor=black](6,6){0.2}
\endpspicture
\hspace*{2cm}
\psset{unit=0.5cm}
\pspicture(-0.5,-0.5)(6.5,6.5)
\pspolygon[fillstyle=solid,fillcolor=cyan,linecolor=white](0,0)(0,6)(6,6)(6,0)
\pspolygon[fillstyle=solid](1.2,1)(1,1.2)(1,4.8)(1.2,5)(1.8,5)(2,4.8)(2,1.2)(1.8,1)
\pspolygon[fillstyle=solid](3.2,4)(3,4.2)(3,5.3)(3.2,5.5)(3.8,5.5)(4,5.3)(4,4.2)(3.8,4)
\pspolygon[fillstyle=solid](2.5,0)(2.5,1.8)(2.7,2)(4.3,2)(4.5,1.8)(4.5,0)
\psline[linecolor=white](2.5,0)(4.5,0)
\pspolygon[fillstyle=solid](5.2,5)(5,5.2)(5,6)(6,6)(6,5)
\psline[linecolor=white](5,6)(6,6)(6,5)
\pspolygon[fillstyle=solid](4.2,2.5)(4,2.7)(4,3.3)(4.2,3.5)(4.8,3.5)(5,3.3)(5,2.7)(4.8,2.5)
\endpspicture
\end{center}
\caption{Left: example of absorbing state with $q=\frac{3}{4}$. Right: example
  with $q=\frac{5}{8}$ (in blue; the grid is not represented)} 
\label{fig:4}
\end{figure}

\paragraph{Other absorbing classes}
 Let us reconsider the notion of trajectory. As the process is
  deterministic, a (deterministic) trajectory is merely a sequence of states
  $X_1,X_2,X_3,\ldots$. We show that cycles (periodic trajectories) and
infinite (aperiodic) trajectories may exist. Figure
\ref{fig:5} presents an example of a cycle, and Figure~\ref{fig:6} an example of
an infinite trajectory.
\begin{figure}[htp]
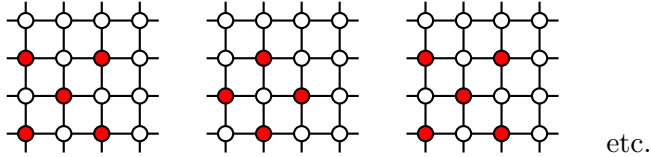

\begin{center}
\psset{unit=0.5cm}
\pspicture(-0.5,-0.5)(3.5,3.5)
\psline(-0.5,0)(3.5,0)
\psline(-0.5,1)(3.5,1)
\psline(-0.5,2)(3.5,2)
\psline(-0.5,3)(3.5,3)
\psline(0,-0.5)(0,3.5)
\psline(1,-0.5)(1,3.5)
\psline(2,-0.5)(2,3.5)
\psline(3,-0.5)(3,3.5)
\pscircle[fillstyle=solid,fillcolor=red](0,0){0.2}
\pscircle[fillstyle=solid](0,1){0.2}
\pscircle[fillstyle=solid,fillcolor=red](0,2){0.2}
\pscircle[fillstyle=solid](0,3){0.2}
\pscircle[fillstyle=solid](1,0){0.2}
\pscircle[fillstyle=solid,fillcolor=red](1,1){0.2}
\pscircle[fillstyle=solid](1,2){0.2}
\pscircle[fillstyle=solid](1,3){0.2}
\pscircle[fillstyle=solid,fillcolor=red](2,0){0.2}
\pscircle[fillstyle=solid](2,1){0.2}
\pscircle[fillstyle=solid,fillcolor=red](2,2){0.2}
\pscircle[fillstyle=solid](2,3){0.2}
\pscircle[fillstyle=solid](3,0){0.2}
\pscircle[fillstyle=solid](3,1){0.2}
\pscircle[fillstyle=solid](3,2){0.2}
\pscircle[fillstyle=solid](3,3){0.2}
\endpspicture
\hspace*{0.5cm}
\pspicture(-0.5,-0.5)(3.5,3.5)
\psline(-0.5,0)(3.5,0)
\psline(-0.5,1)(3.5,1)
\psline(-0.5,2)(3.5,2)
\psline(-0.5,3)(3.5,3)
\psline(0,-0.5)(0,3.5)
\psline(1,-0.5)(1,3.5)
\psline(2,-0.5)(2,3.5)
\psline(3,-0.5)(3,3.5)
\pscircle[fillstyle=solid](0,0){0.2}
\pscircle[fillstyle=solid,fillcolor=red](0,1){0.2}
\pscircle[fillstyle=solid](0,2){0.2}
\pscircle[fillstyle=solid](0,3){0.2}
\pscircle[fillstyle=solid,fillcolor=red](1,0){0.2}
\pscircle[fillstyle=solid](1,1){0.2}
\pscircle[fillstyle=solid,fillcolor=red](1,2){0.2}
\pscircle[fillstyle=solid](1,3){0.2}
\pscircle[fillstyle=solid](2,0){0.2}
\pscircle[fillstyle=solid,fillcolor=red](2,1){0.2}
\pscircle[fillstyle=solid](2,2){0.2}
\pscircle[fillstyle=solid](2,3){0.2}
\pscircle[fillstyle=solid](3,0){0.2}
\pscircle[fillstyle=solid](3,1){0.2}
\pscircle[fillstyle=solid](3,2){0.2}
\pscircle[fillstyle=solid](3,3){0.2}
\endpspicture
\hspace*{0.5cm}
\pspicture(-0.5,-0.5)(3.5,3.5)
\psline(-0.5,0)(3.5,0)
\psline(-0.5,1)(3.5,1)
\psline(-0.5,2)(3.5,2)
\psline(-0.5,3)(3.5,3)
\psline(0,-0.5)(0,3.5)
\psline(1,-0.5)(1,3.5)
\psline(2,-0.5)(2,3.5)
\psline(3,-0.5)(3,3.5)
\pscircle[fillstyle=solid,fillcolor=red](0,0){0.2}
\pscircle[fillstyle=solid](0,1){0.2}
\pscircle[fillstyle=solid,fillcolor=red](0,2){0.2}
\pscircle[fillstyle=solid](0,3){0.2}
\pscircle[fillstyle=solid](1,0){0.2}
\pscircle[fillstyle=solid,fillcolor=red](1,1){0.2}
\pscircle[fillstyle=solid](1,2){0.2}
\pscircle[fillstyle=solid](1,3){0.2}
\pscircle[fillstyle=solid,fillcolor=red](2,0){0.2}
\pscircle[fillstyle=solid](2,1){0.2}
\pscircle[fillstyle=solid,fillcolor=red](2,2){0.2}
\pscircle[fillstyle=solid](2,3){0.2}
\pscircle[fillstyle=solid](3,0){0.2}
\pscircle[fillstyle=solid](3,1){0.2}
\pscircle[fillstyle=solid](3,2){0.2}
\pscircle[fillstyle=solid](3,3){0.2}
\endpspicture
\hspace*{0.5cm} etc.
\end{center}
\caption{Example of a cycle with $\ZZ^2$ and the 1-neighborhood ($q=\frac{1}{2}$)} 
\label{fig:5}
\end{figure}

\begin{figure}[htp]
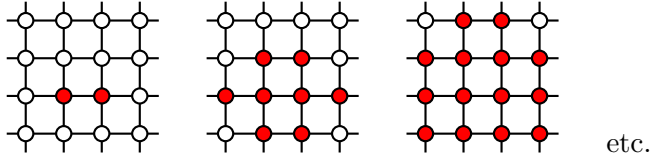

\begin{center}
\psset{unit=0.5cm}
\pspicture(-0.5,-0.5)(3.5,3.5)
\psline(-0.5,0)(3.5,0)
\psline(-0.5,1)(3.5,1)
\psline(-0.5,2)(3.5,2)
\psline(-0.5,3)(3.5,3)
\psline(0,-0.5)(0,3.5)
\psline(1,-0.5)(1,3.5)
\psline(2,-0.5)(2,3.5)
\psline(3,-0.5)(3,3.5)
\pscircle[fillstyle=solid](0,0){0.2}
\pscircle[fillstyle=solid](0,1){0.2}
\pscircle[fillstyle=solid](0,2){0.2}
\pscircle[fillstyle=solid](0,3){0.2}
\pscircle[fillstyle=solid](1,0){0.2}
\pscircle[fillstyle=solid,fillcolor=red](1,1){0.2}
\pscircle[fillstyle=solid](1,2){0.2}
\pscircle[fillstyle=solid](1,3){0.2}
\pscircle[fillstyle=solid](2,0){0.2}
\pscircle[fillstyle=solid,fillcolor=red](2,1){0.2}
\pscircle[fillstyle=solid](2,2){0.2}
\pscircle[fillstyle=solid](2,3){0.2}
\pscircle[fillstyle=solid](3,0){0.2}
\pscircle[fillstyle=solid](3,1){0.2}
\pscircle[fillstyle=solid](3,2){0.2}
\pscircle[fillstyle=solid](3,3){0.2}
\endpspicture
\hspace*{0.5cm}
\pspicture(-0.5,-0.5)(3.5,3.5)
\psline(-0.5,0)(3.5,0)
\psline(-0.5,1)(3.5,1)
\psline(-0.5,2)(3.5,2)
\psline(-0.5,3)(3.5,3)
\psline(0,-0.5)(0,3.5)
\psline(1,-0.5)(1,3.5)
\psline(2,-0.5)(2,3.5)
\psline(3,-0.5)(3,3.5)
\pscircle[fillstyle=solid](0,0){0.2}
\pscircle[fillstyle=solid,fillcolor=red](0,1){0.2}
\pscircle[fillstyle=solid](0,2){0.2}
\pscircle[fillstyle=solid](0,3){0.2}
\pscircle[fillstyle=solid,fillcolor=red](1,0){0.2}
\pscircle[fillstyle=solid,fillcolor=red](1,1){0.2}
\pscircle[fillstyle=solid,fillcolor=red](1,2){0.2}
\pscircle[fillstyle=solid](1,3){0.2}
\pscircle[fillstyle=solid,fillcolor=red](2,0){0.2}
\pscircle[fillstyle=solid,fillcolor=red](2,1){0.2}
\pscircle[fillstyle=solid,fillcolor=red](2,2){0.2}
\pscircle[fillstyle=solid](2,3){0.2}
\pscircle[fillstyle=solid](3,0){0.2}
\pscircle[fillstyle=solid,fillcolor=red](3,1){0.2}
\pscircle[fillstyle=solid](3,2){0.2}
\pscircle[fillstyle=solid](3,3){0.2}
\endpspicture
\hspace*{0.5cm}
\pspicture(-0.5,-0.5)(3.5,3.5)
\psline(-0.5,0)(3.5,0)
\psline(-0.5,1)(3.5,1)
\psline(-0.5,2)(3.5,2)
\psline(-0.5,3)(3.5,3)
\psline(0,-0.5)(0,3.5)
\psline(1,-0.5)(1,3.5)
\psline(2,-0.5)(2,3.5)
\psline(3,-0.5)(3,3.5)
\pscircle[fillstyle=solid,fillcolor=red](0,0){0.2}
\pscircle[fillstyle=solid,fillcolor=red](0,1){0.2}
\pscircle[fillstyle=solid,fillcolor=red](0,2){0.2}
\pscircle[fillstyle=solid](0,3){0.2}
\pscircle[fillstyle=solid,fillcolor=red](1,0){0.2}
\pscircle[fillstyle=solid,fillcolor=red](1,1){0.2}
\pscircle[fillstyle=solid,fillcolor=red](1,2){0.2}
\pscircle[fillstyle=solid,fillcolor=red](1,3){0.2}
\pscircle[fillstyle=solid,fillcolor=red](2,0){0.2}
\pscircle[fillstyle=solid,fillcolor=red](2,1){0.2}
\pscircle[fillstyle=solid,fillcolor=red](2,2){0.2}
\pscircle[fillstyle=solid,fillcolor=red](2,3){0.2}
\pscircle[fillstyle=solid,fillcolor=red](3,0){0.2}
\pscircle[fillstyle=solid,fillcolor=red](3,1){0.2}
\pscircle[fillstyle=solid,fillcolor=red](3,2){0.2}
\pscircle[fillstyle=solid](3,3){0.2}
\endpspicture
\hspace*{0.5cm} etc.
\end{center}
\caption{Example of an infinite trajectory with $\ZZ^2$ and the 1-neighborhood
  ($q=\frac{1}{4}$). The initial $X$ constantly grows.}
\label{fig:6}
\end{figure}


\begin{proposition}
Suppose $A$ is an anonymous and Boolean aggregation function. Then absorbing
classes are either singletons $\{X\}$, where $X\in 2^{\scX}$, or cycles (periodic
trajectories) of nonempty sets $\{X_1,\ldots,X_k\}$ with the condition that all
sets are pairwise incomparable by inclusion.
\end{proposition}
\begin{proof}
Every transition being with probability 1, the evolution is a deterministic trajectory, hence a sequence of configurations $X_0,X_1,\ldots,X_k,\ldots$. Note that if a repetition occurs, then the trajectory is periodic. The following cases are exhaustive and exclusive:
\begin{enumerate}
\item Fixed points, $X_{k+1}=X_k$ for $k\geq K$ exist, and are absorbing states.
\item Cycles (periodic trajectories) exist, and are periodic absorbing classes.
\item Otherwise, the trajectory does not have any repetition and is
  infinite. But this cannot be a class, because there would be no transition
  from state $X_{k+1}$ to $X_k$.
\end{enumerate} 
\end{proof}

 The results show that, contrarily to the probabilistic model with strict
  aggregation functions, polarization might arise: there could exist
  finite or infinite sets of agents being active or inactive forever, provided
  these sets satisfy some properties regarding their ``shape''. Cycles may also occur.

\section{Concluding remarks}\label{sec:conclusion}

Our main question was how the diffusion evolves from a finite set of active agents and more generally from any set of active agents. The main finding is that the answer depends on the aggregation function for some properties and on the structure of the network for others. We first showed that the transience/persistence of a state relies only on the type of diffusion
mechanism, which in our model amounts to the mechanism of aggregating the statuses of the neighbors, without any further restriction on the network. On the contrary, the possibility to go from one configuration to another one inside a class (irreducibility) is closely related to the structure of the graph. We provide a mild sufficient condition on the structure, called richness, which permits to obtain irreducibility. 

Among diffusion mechanisms, we clearly establish a distinction between the probabilistic and the deterministic mechanism, the latter being nothing other than the threshold model, studied by \cite{mor00}. In the former, we have supposed in our analysis that the probability of being active strictly increases with the number of active neighbors. Under this assumption, we have shown that no polarization can occur, even in the weak sense: no set of  active agents can remain stable, even if we allow some variation around this set. On the contrary, the diffusion is erratic and homogeneous on the whole   network, and does not fix on some peculiar region of it. By contrast, the deterministic model allows the appearance of stable finite or infinite sets of
  active/inactive agents, that is, polarization can appear, and under many  different forms.

\vspace{5mm}

The next step is to study nonstrict nonBoolean aggregation function, which reveals to be pretty tricky. To highlight the difficulties, let us  assume for simplicity that the aggregation function $A$ is anonymous. $A$
being nonstrict and nonBoolean implies that there exist
$\ell,r\in\{0,\ldots,\gamma-1\}$ such that
\begin{equation}\label{eq:4}
A(1_S)=\begin{cases}
0 & \text{iff }|S|\leq \ell\\
1 & \text{iff }|S|\geq \gamma-r,
\end{cases}
\end{equation}
with $\ell\vee r>0$ and $\ell+r<\gamma$. This case is more complex than the two previous ones. Indeed, it has the aspect both of the Boolean case and of the strict aggregation case.
For example, take $\ZZ^2$ with the 1-neighborhood and the following anonymous aggregation function
\begin{equation}\label{eq:5b}
A(1_S)=\begin{cases}
0 & \text{iff }|S|=0\\
1/2 & \text{iff }|S|=1,\\
1 & \text{iff }|S|\geq 2.
\end{cases}
\end{equation}

We see an asymmetry between the role of $0$ and $1$.

Proposition \ref{persistence} that focuses on the persistence of the number of elements is only partially true here: the two first bullet points concerning the case with no active agents or the case with a finite number of active agents are still valid. On the contrary, the third one concerning the case with an infinite number of agents is not true anymore, as shown by the following configuration. Let $\omega$ be such that
\[
\forall (n,m)\in \ZZ^2,\ \omega(n,m)=0 \text{ if and only if } n+m\in 10\ZZ.
\]
The transition is deterministic and leads to $\textbf{1}$ since the inactive agents are all isolated. Proposition \ref{disappearing} that proves the possibility for a finite number of active agents to disappear is not true anymore, since a square of active agents will never disappear. On the contrary, with positive probability a finite set of inactive agents may disappear. Lemma \ref{placement} that studies the possibility for a trajectory of diffusion to exist between two configurations seems to be true for active agents but not true for inactive agents.\\

This makes the study of nonstrict nonBoolean aggregation function pretty tricky, and we leave it for further research. It will need to disentangle the condition for each of the lemmas developed for strict aggregation functions and boolean aggregation functions.

\section*{Acknowledgment} The authors thank Matthew Jackson, Ron Peretz 
  and Marcus Pivato for fruitful discussions.

\section*{Funding}
M.~Grabisch and A.~Rusinowska acknowledge the financial support by European
Union’s Horizon 2020 research and innovation programme under the Marie
Sk{\l}odowska-Curie grant agreement No 721846, “Expectations and Social
Influence Dynamics in Economics (ExSIDE), and the National Agency for Research,
Project DynaMITE (ANR-13-BSH1-0010-01). X.~Venel acknowledges the financial
support by the National Agency for Research, Project CIGNE
(ANR-15-CE38-0007-01).

\section*{Conflict of interest}
The authors declare that they have no conflict of interest.


\bibliographystyle{plainnat}
\bibliography{references,rusinowska}

\appendix

\section{Proof of intermediate results}\label{appA}

\paragraph{Proof of Lemma~\ref{lem:kernel}.}

 By construction, $K(\omega,\cdot)$ satisfies the first hypothesis of a
 kernel. We now check that given $\mathcal{A} \in \mathcal{T}$, the function
 $K(\cdot,\mathcal{A})$ is measurable. We first consider a set $\mathcal{A}$ in
 the basis of $\mathcal{T}$. Hence, it is a finite cylinder, i.e., there exist
 two finite disjoint sets $X,Y$ such that $\mathcal{A}=(X,Y)^+$. Then we know
 that $K(\omega,(X,Y)^+)$ depends only on the restriction of $\omega$ to agents
 in $\overline{X\cup Y}$, hence on a finite set of agents. By definition of the
 product topology, the function is continuous and therefore measurable. Let
\[
\mathcal{H}=\{\mathcal{A} \in \mathcal{T}, \ \omega\rightarrow K(\omega,\mathcal{A}) \text{ is measurable} \} \subset \mathcal{T}
\]
be the set of sets in $\mathcal{T}$ such that the mapping is measurable. Then $\mathcal{H}$ is a $\sigma$-field since:
\begin{itemize}
\item[-] it is nonempty: $\varnothing \in \mathcal{H}$;
\item[-] it is stable under complementation since $K(\omega,\overline{\mathcal{A}})=1-K(\omega,\mathcal{A})$;
\item[-] it is stable under countable unions: Let $(\mathcal{A}_n)_{n\in \NN}$, then
\begin{align*}
K(\omega, \cup_{n\in \NN} \mathcal{A}_{n}) &= \sum_{n \in \NN} K(\omega,\mathcal{A}_n).
\end{align*}
The countable sum of measurable mappings is also measurable.
\end{itemize}
Since $\mathcal{H}$ contains the finite cylinders, it follows that
$\mathcal{T} \subset \mathcal{H}$ and we have the equality: for every $\mathcal{A}
\in \mathcal{T}$, the function $K(\cdot,\mathcal{A})$ is measurable.

\paragraph{Proof of Proposition~\ref{pro:3}.}

We first observe that $\intX=\cloX\neq X$ is equivalent to the two conditions:
\begin{itemize}
\item Any $x\in\scX$ has all its neighbors in $X$ or none of them.\hfill $(*)$
\item Either there exists $x\not\in X$ which has (all) neighbors in $X$ or there
  exists $x\in X$ with no neighbor in $X$. \hfill $(**)$
\end{itemize}
\begin{enumerate}
\item The two conditions $(*)$ and $(**)$ are invariant to the change
  $X\rightarrow X^c$.
\item Assume by $(**)$ that there exists $x\in X$ with no neighbor in $X$ (the
  other case works by (*)). We claim that $X\cap \cloX=\varnothing$. Indeed,
  suppose not and consider $y\in X\cap \cloX$. By connectedness, there exists a
  finite path $x_1=x,x_2,\ldots,x_k=y$ connecting $x$ to $y$. However,
  $x_2\in\Gamma(x_1)$, therefore $x_2\in X^c$. By symmetry of $\sim$,
  $\Gamma(x_2)\ni x_1$, which by (*) implies that $\Gamma(x_2)\subseteq X$,
  i.e., $x_2\in \cloX$. Therefore, $x_3\in X$, but $x_3\not\in \cloX$ because by
  symmetry $X^c\ni x_2\in \Gamma(x_3)$. Hence, $x_\ell$ with odd index $\ell$
  belongs to $X\cap (\cloX)^c$ while those with even indices belong to $X^c\cap
  \cloX$. Therefore, $y$ can never be reached, a contradiction.

By the same argument, it follows that
$X^c\cap(\cloX)^c=\varnothing$. Consequently, $X^c=\cloX$.
\end{enumerate}

\paragraph{Proof of Proposition~\ref{no_star_graph}.}
 The ``if'' part is obvious. As for the ``only if'' part, let us first assume
 that there exists a double ray in the graph, i.e., there exists a sequence
 $(x_m)_{m\in \ZZ}$ in $\scX^{\ZZ}$ such that all nodes are different and for
 every $n \in \ZZ$, $\{x_{n-1},x_{n+1}\} \subseteq \Gamma(x_n)$.

Moreover, since any node on the double ray admits already two branches of size $2$, the absence of complex star implies that:
\begin{itemize}
\item any additional node linked to the ray has no other neighbor. We will say that it forms an antenna.
\item for every distinct $i,j\in \ZZ$ which are not consecutive there is
  no link between $x_i$ and $x_j$, since similarly this would generate a complex
  star with $x_i$ as a center. Therefore, $\Gamma(x_n)\cap (x_m)_{m\in
      \ZZ}=\{x_{n-1},x_{n+1}\}$. 
\end{itemize}
Hence, we obtain the second case in our characterization of networks without complex stars.

\smallskip

Let us suppose now that there exists no double ray. Since $(\scX,\sim)$ is an
infinite connected graph with bounded degree, we know by Proposition 8.2.1. from
\cite{die05} that there exists a ray, i.e., there exists $(x_i)_{i\geq 0}$ such
that $x_1\in \Gamma(x_0)$ and
\[
\forall i\geq 1, \{x_{i-1},x_{i+1}\} \subseteq \Gamma(x_i).
\]
The preceding reasoning on additional nodes can be applied for every $i\geq 2$
(it is unknown if $x_0$ or $x_1$ have more than one neighbors): any node
not in $(x_i)_{i\geq 0}$ and connected to $(x_i)_{i\geq 2}$ has only one neighbor.

\medskip

Denote by $\scX_1=\cup_{i\geq 2}(\{x_i\} \cup \Gamma(x_i))$ and consider the
graph with the set of nodes $\scX_2=(\scX \setminus\scX_1)\cup\{x_1\}$ and the relation $\sim$. This new graph is still connected. Moreover, it is finite,
otherwise there would exist a ray which could be connected to the previous ray
to generate a double ray. Let us consider a node $x'$ at maximal distance from
$x_1$ in the subgraph $(\scX_2,\sim)$. We define the new sequence $(x'_n)_{n\geq
  0}$ with $x'_0=x'$, made from the path from $x'$ to $x_1$ concatenated with the ray $(x_n)_{n\geq 1}$.
\begin{itemize}
\item The reasoning on antennas can now be applied to any node in $\{x'_n,n\geq
  2\}$.
\item  Moreover, every neighbor of $x'_1$ (except $x'_2$) has only
  $x'_1$ as neighbor, otherwise $x'$ would not be a farthest node from $x_1$.
\item Finally, there cannot be any internal link in the ray $(x'_n)_{n\geq 0}$,
  since any such link has at least one extremity in $\{x'_n,n\geq 2\}$ and would
  create a complex star.
\end{itemize}

\paragraph{Proof of Proposition~\ref{prop:kreg}.}
  Due to the previous discussion, it suffices to prove that for any partial
  configuration $(X,Y)$ on $\scX_o$ (or $\scX_e$), the bipartite graph $(X\cup
  Y,\overline{X\cup Y},\sim)$ admits a matching of $X\cup Y$. We consider first
  the case of a $k$-regular graph. Observe that by assumption, the degree of
  each node of $X\cup Y$ in the bipartite graph is $k$, while the degree of each
  node of $\overline{X\cup Y}$ in the bipartite graph is at most $k$ (some of
  these nodes have necessarily neighbors outside $X\cup Y$).  Consider
  $S\subseteq X\cup Y$ and the set of neighbors of $S$, denoted by
  $\Gamma(S):=\bigcup_{x\in S}\Gamma(x)$. By the above observation, the number
  of edges from $S$ to $\Gamma(S)$ is $k|S|$, and the total number of edges
  arriving at $\Gamma(S)$ is at most $k|\Gamma(S)|$. Therefore $|\Gamma(S)|\geq
  |S|$. Then, by Hall's Theorem (see, e.g., \cite[Th. 2.1.2]{die05}), a matching
  of $X\cup Y$ exists.

  Suppose now that $(\scX,\sim)$ is a hierarchy. We know that $\scX_o$ is the
  set of nodes in the odd layers of the hierarchy. Take $x\in X\cup Y$. It
  suffices to define $\theta(x)$ as one of the subordinates of $x$,    which lies on the next even layer. This will form
  a matching  of $X\cup Y$ as two distinct nodes in $X\cup Y$ have
    necessarily disjoint sets of subordinates.

\section{Proof of Theorem~\ref{absorbing-transient}}\label{appB}
We have only to prove (\ref{null}) and
(\ref{transient}). In Section \ref{sec:pers}, we focus on the ``number" of
inactive/active nodes on even/odd positions and prove that the sets containing
only a null or infinite number of each element are absorbing. Then, in
Section~\ref{sec:trans}, we prove that the sets containing finitely many
elements are transient.

\subsection{Persistence of the ``number'' of elements}\label{sec:pers}

The following proposition states that the ``number'' of active agents on even
(respectively odd) positions yields the same ``number'' of active agents on odd
(respectively even) positions at the next stage. We establish the result for the
case of active agents on even positions. The same result holds if active is replaced by inactive and/or $\mathcal{X}_e$ by $\mathcal{X}_o$.

\begin{proposition}\label{persistence}
For every $\omega \in \Omega$, let $\omega'$ be a configuration randomly chosen
according to the kernel $K(\omega,\cdot)$. Then:
\begin{itemize}
\item if $\{x \in \mathcal{X}_e\mid \omega(x)=1\}$ is empty then $\omega'$ has almost surely no  active agents on odd positions;
\item if $\{x \in \mathcal{X}_e\mid \omega(x)=1\}$ is finite then $\omega'$ has almost surely a finite (maybe zero) number of active agents on odd positions;
\item if $\{x \in \mathcal{X}_e\mid \omega(x)=1\}$ is infinite then $\omega'$ has almost surely an infinite number of active agents on odd positions.
\end{itemize}
\end{proposition}

\begin{proof}
Denote by $\Theta=\{x \in \mathcal{X}_e\mid \omega(x)=1\}$ the set of active agents on an even position.

First, it is clear that if this set is empty, then every odd position is surrounded by only inactive agents and therefore is inactive with probability $1$ at the next stage. Hence, there is  almost surely no active node on odd positions in $\omega'$.

Secondly, let $x$ be an odd position. Since the aggregation function is strict,
the probability that $\omega'(x)$ is equal to $1$ is strictly positive, if and
only if, there exists a neighbor $x'$ of $x$ such that $\omega(x')=1$. The set
of neighbors of odd positions is the set of even positions and each agent has at
most $\gamma$ neighbors, therefore almost surely 
\[
|\{x \in \mathcal{X}_o \mid \omega'(x)=1\}| \leq \gamma |\{x\in \mathcal{X}_e \mid\omega(x)=1 
\}|.
\]
Hence, if the right-hand side is finite, so is the left-hand side.

Finally, we assume that $\{x \in \mathcal{X}_e\mid \omega(x)=1\}$ is infinite. We compute the probability that $\omega'$ has a finite number of active nodes on odd states. If this is the case, it implies that there exists an infinite number of elements  $x\in \Theta$ such that
\[
\forall x'\in \Gamma(x),\ \omega'(x')=0.
\]

But by assumption, $x'$ has at least one neighbor which is active, it follows that the probability that $\omega'(x')=0$ is strictly smaller than $1$. Since the probabilities are independent, it follows that the total probability of the event is $0$.
\end{proof}

 Theorem~\ref{absorbing-transient}.(\ref{null}) is then an immediate consequence of the above proposition.

\subsection{Transient set}\label{sec:trans}
We now prove that if there is a strictly positive but finite number of active agents at odd positions, then there is a positive probability to reach a configuration where there are no active agents anymore on odd positions. Similar results hold for inactive and/or for odd positions. This result can be expressed in terms of reachability.


\begin{lemma}\label{disappearing_lemma}
Let $X\subset \scX_o$.
\begin{itemize}
\item[-] If $X$ is a finite set, then $\int(\int (X))\subset X$.
\item[-] If $X$ is a co-finite set, then $X \subset \clo(\clo (X)))$.
\end{itemize}
\end{lemma}

\begin{proof}
We prove the first result. The second one can be deduced immediately by complementation.

\smallskip

Let us prove the result by contradiction. First, we know that $\int(\int(X))
\subset \scX_o$. Take $x\in \int(\int(X))$ and assume that $x\not\in X$. By
definition of $\scX_o$ and $\scX_e$, we know that $\int(X) \cap
\scX_o=\varnothing$, hence $\Gamma(x)\not\subseteq X$, i.e., there exists $y\in
\Gamma(x)$ s.t. $y\not\in X$. Since $x\in \int(\int(X))$, $\Gamma(x)\subseteq
\int(X)$ and therefore $y \in \int(X)$. It follows that $\Gamma(y)\subseteq
X$. However, by symmetry of $\sim$, we have $x\in \Gamma(y)$, so $x\in X$, a
contradiction. This proves that $\int(\intX)\subseteq X$.

\smallskip

Suppose that for all $x\in X$, $x\in\int(\int(X))$. It follows that for all $x\in X$, $\Gamma(x)\subseteq \int(X)$, hence
\begin{align}\label{chain}
\forall x\in X, \ \forall y\in \Gamma(x), \forall z \in \Gamma(y), z\in X.
\end{align}
Since $X$ is finite, pick some $t\in \scX_o \cap X^{c}$. By connectedness of $\sim$, there should exist a finite path $x_0=x,x_2,\ldots,x_{2k}=t$ from $x\in
X$ to $t \notin X$. However, no such path can exist by Equation (\ref{chain}).

Therefore, there exists $x\in X$ such that $x\not\in\int(\int(X))$. This proves that the inclusion is strict.
\end{proof}

We can immediately deduce from Lemma \ref{disappearing_lemma} the following proposition.

\begin{proposition}\label{disappearing}
For every $\omega \in \Omega$, if $\{x \in \mathcal{X}_o\mid \omega(x)=1\}$ is finite then the set of configurations with only $0$ on odd positions is reachable.
\end{proposition}

\begin{proof}
As $X:=\{x \in \mathcal{X}_o\mid \omega(x)=1\}$ is finite, Lemma~\ref{disappearing_lemma} implies that the sequence $\int^2(X),
\int^4(X), \ldots$ converges to $\varnothing$ in a finite
number of steps. Since $A$ is a strict aggregation rule, for every $n\geq 0$, the probability to reach $\int^{2n+2}(X)$ from $\int^{2n}(X)$ is positive, hence we constructed a trajectory between $X$ and $\varnothing$ with even stages: hence a trajectory from a configuration in $(.,.,.,\F)$ to a configuration in $(.,.,.,0)$. 
\end{proof}

\begin{remark}
Let us stress out that this lemma does not say anything about the presence of active agents at an even stage on an even position. 
\end{remark}

We can prove in a similar way the same result for $X\subset \scX_e$. This yields the proof of  Theorem~\ref{absorbing-transient}.(\ref{transient}).

\section{Proof of Theorem~\ref{irreducibility}}\label{appC}
\noindent In order to prove Theorem \ref{irreducibility}, we establish the following two propositions. The theorem is then an immediate consequence.

\begin{proposition}\label{cylinder_odd}
Assume that $(\scX,\sim)$ satisfies the Richness Assumption. Let $\omega$ be a configuration such that there exist both $0$ and $1$ on an odd position. 
\begin{itemize}
\item[-] Let $X\subset \mathcal{X}_e$ and $Y\subset \mathcal{X}_e$ be two finite sets, then the cylinder   $(X,Y)^+$ is reachable from $\omega$ in an odd number of stages.
\item[-] Let $X\subset \mathcal{X}_o$ and $Y\subset \mathcal{X}_o$ be two finite sets, then the cylinder   $(X,Y)^+$ is reachable from $\omega$ in an even number of stages.
\end{itemize}
\end{proposition}

\begin{proposition}\label{cylinder_even}
Assume that $(\scX,\sim)$ satisfies the Richness Assumption. Let $\omega$ be a configuration such that there exist both $0$ and $1$ on an even  position. 
\begin{itemize}
\item[-] Let $X\subset \mathcal{X}_e$ and $Y\subset \mathcal{X}_e$ be two finite sets, then the cylinder   $(X,Y)^+$ is reachable from $\omega$ in an even number of stages.
\item[-] Let $X\subset \mathcal{X}_o$ and $Y\subset \mathcal{X}_o$ be two finite sets, then the cylinder   $(X,Y)^+$ is reachable from $\omega$ in an odd number of stages.
\end{itemize}
\end{proposition}

In Propositions \ref{cylinder_odd} and \ref{cylinder_even} we restricted ourselves to finite cylinders. We can extend these results by using the fact that $\Omega$ is a polish set. Let $\cA$ be such that $\phi(\cA)>0$. $\phi$ is a Borel measure and therefore is tight:
\[
\forall  \cA \in \mathcal{T},\ \phi(\cA)=\sup_{K \text{ compact }\subset A} \phi(K).
\]
It follows that there exists a closed set $K$ such  that $\phi(K)>0$. Any open set can be written as a countable union of finite cylinders, and therefore any closed set can be written as a countable intersection of complements of finite cylinders. However, the complement of a finite cylinder is a union of finite cylinders:
\[
((X,Y)^+)^c=\left(\bigcup_{x\in X} (\varnothing,\{x\})^+ \right) \cup \left(\bigcup_{y\in Y} (\{y\},\varnothing)^+ \right).
\]
Hence, a closed set can be written as a countable union of countable intersections of finite cylinders, and therefore one of these intersections has a non-zero measure under $\phi$: 
there exists $(X_n,Y_n)_{n\in \NN}$ a sequence of finite cylinders such that
\[
\phi \left(\bigcap_{n\in \NN} (X_n,Y_n)^+ \right) >0.
\]
This is only possible if $X_\infty=\bigcup_{n\in \NN} X_n$ and $Y_\infty=\bigcup_{n\in \NN} Y_n$ are finite and disjoint. In this case, we have
\[
\bigcap_{n\in \NN} (X_n,Y_n)^+=(X_\infty,Y_\infty)^+.
\]
$(X_\infty,Y_\infty)^+$ is reachable from $\omega$ and therefore $\cA$ is reachable from $\omega$.

In the next section, we will prove Proposition \ref{cylinder_odd}. Proposition \ref{cylinder_even} can be proven in a very similar way. The proof relies on constructing an explicit trajectory between the initial partial configuration $\omega$ to any other partial configuration $\omega'$ compatible with the condition in the lemma. In order to do so, the intuition is to fix the status of each agent in the network one by one. In doing so, one needs to be careful about not erasing a status that has been fixed before in the procedure.

\subsection{Proof of Proposition \ref{cylinder_odd}}

The exposition of the proof is divided into several parts. We first present two partial mechanisms of transitions. The first one is a storing mechanism of partial configuration conditionally on the existence of a storing function. The second one is the propagation of an active status or the propagation of an inactive status in the network. These two procedures are defined locally, i.e., on partial configurations.

Then, we will combine these partial mechanisms into a global well-defined procedure. Hence, it will induce a trajectory. The key point is to ensure that these elementary mechanisms are not conflicting each other by using the same nodes at the same time.

We restrict ourselves to what happens on $\mathcal{X}_o$. But any of the result
can be expressed on $\mathcal{X}_e$.  The bipartite structure of $(\scX,\sim)$ ensures that one can decompose the dynamic into two parts: the status of odd stages/odd positions and even stages/even positions on one side and the status of even stages/odd positions and odd stages/even positions on the other side.

\subsubsection{Storing and propagation}
We begin by remarking that one can store the status of a node on one of its neighbor. Let $\omega \in \Omega$ be a configuration such that agent
$x$ in $\mathcal{X}_o$ is active, and $x'$ be one of its neighbors. The
following sequence of cylinders does the job:
\begin{itemize}
\item for every $l\in 2\NN$, $\cA_l=(\{x\},\varnothing)^+$,
\item for every $l\in 2\NN+1$, $\cA_l=(\{x'\},\varnothing)^+$.
\end{itemize}

The next mechanism for storing partial configurations is merely a generalization of the previous idea.

\begin{lemma}\label{global_saving}
If a partial configuration $(X,Y)$ of $\mathcal{X}_o$ can be stored then for every $n\in 2\NN$,
\[
\forall \omega \in (X,Y)^+,\ K^n(\omega,(X,Y)^+)>0.
\]
\end{lemma}
 
The procedure used in the proof will be called ``storing a partial configuration in the global procedure''.

\begin{proof}
We assume that the finite configuration $(X,Y)$ of $\scX_o$ can be stored. Hence, there exists $\theta$ such that for all $x\in X\cup Y$, $\theta(x) \in \Gamma(x)$ and for all $x \in X$, for all $y\in Y$, $\theta(x) \neq \theta(y)$. Define $\theta(X)=\{\theta(x),\ x\in X\}$ and $\theta(Y)=\{\theta(y),\ y\in Y\}$
We consider the following sequence of cylinders:
\begin{itemize}
\item for every $l\in 2\NN$, $\cA_l=(X,Y)^+$,
\item for every $l\in 2\NN+1$, $\cA_l=(\theta(X),\theta(Y))^+$.
\end{itemize}
By construction, $\theta(X) \subseteq \overline{X}$ and $\theta(Y) \subseteq \overline{Y}$. Moreover, by symmetry of the neighborhood relation, $X \subseteq \overline{\theta(X)}$ and $Y \subseteq \overline{\theta(Y)}$. By definition of $\theta$ we know that $\theta(X)\cap \theta(Y) =\varnothing$. Hence, this sequence is well defined and 
\[
\forall m\in \cA_l, \ K(m,\cA_{l+1})>0.
\]
Moreover, this can be ensured by only imposing restriction on the status of agents in $X\cup Y$ at even stages and in $\overline{X \cup Y}$ at odd stages.
\end{proof}

Several storing procedures can be run at the same time if they operate on disjoint sets of nodes with disjoint neighborhood.

\medskip

We now turn to the problem of propagating the status of a node in a partial configuration to any node in the network. This procedure will be called ``propagating in the global procedure''.

\begin{lemma}\label{placement}
Let $(X,Y)$ be a partial configuration of $\scX_o$ such that $X$ is not
empty. Then, for any $\omega\in(X,Y)^+$ and for every $x' \in \scX$,
\begin{itemize}
\item if $x' \in \mathcal{X}_e$  then there exists $n\in 2\NN$ such that
\[
K^n(\omega,(\{x'\},\varnothing)^+)>0,
\]
\item if $x' \in \mathcal{X}_o$ then there exists $n\in 2\NN+1$ such that
\[
K^n(\omega,(\{x'\},\varnothing)^+)>0.
\]
\end{itemize}
\end{lemma}

\begin{proof}
Let $(X,Y)$ be a partial configuration of $\scX_o$ such that $X$ is not empty. Let $x\in X$ be an active agent in $\mathcal{X}_o$ and $x'$ be another agent. Consider a path $x_0,\ldots,x_n$ between $x=x_0$ and $x'=x_n$. Informally, we propagate the status of $x$ to $x'$ along this path. More formally, consider the following sequence of cylinders: 
\[
\forall  l\in \{0,\ldots,n\},\ \cA_l=(\{x_l\},\varnothing)^+.
\]
This sequence is well defined, since for every $l\in \{0,\ldots,n-1\}$, $x_{l+1}\in \Gamma(x_l)$, and we have by construction
\[
\forall m\in \cA_l, \ K(m,\cA_{l+1})>0.
\]
We obtain that a configuration where $x'$ is active is reachable and moreover, depending if $x'$ is even or odd, we obtain the condition on the number of iterations.
\end{proof}
A similar procedure can be described for inactive agents.
%
%


A key point in the previous construction is that the propagation requires a
constraint at stage $l$ only on the status of $x_l$. All other agents can be
indifferently active or not. This will allow us to combine the procedure with
other procedures without having them interfering with each other.

%
%

\subsubsection{Global procedure}

We now turn to establish the general procedure. Let $X\subset \mathcal{X}_o$ and
$Y\subset \mathcal{X}_o$ be two disjoint finite sets. Since the graph satisfies
the Richness Assumption, we know that $(X,Y)$ can be stored. Moreover, there exists an infinite number of complex stars, therefore there exists a complex star such that all nodes of the star are disjoint from $X\cup \overline{X} \cup Y \cup \overline{Y}$. Denote by $(s_*,s_1,s_2,s_3,s'_1,s'_2,s'_3)$ the nodes of this complex star, we assume that $s_*\in \mathcal{X}_e$. Hence, $s_1,s_2,s_3\in \mathcal{X}_o$ and $s'_1,s'_2,s'_3 \in \mathcal{X}_e$.
We also define a distance on the graph as the length of a shortest path between a node and $s_*$:
\[
\forall x\in \mathcal{X}, d(x,s_*)=\inf \{n \in \NN, \exists (x_l)_{1 \leq l\leq n} \text{ a path between } x \text{ and } s_*\}.
\]
 This is a well-defined distance since the graph is connected and there exists
 at least one path between $s_*$ and $x$. We define $\mathcal{X}_{n}=\{x \in
 \mathcal{X}, d(s_*,x) \leq n\}$ and $\mathcal{X}_{n,e}=\mathcal{X}_n \cap
 \mathcal{X}_e$ and $\mathcal{X}_{n,o}=\mathcal{X}_n \cap \mathcal{X}_o$.

 \medskip
 
Let $\omega$ be a configuration with at least one active agent and one
inactive agent in $\mathcal{X}_o$. Let us prove that there exists a trajectory with an odd number of steps from $\omega$ to $(X,Y)^+$.

\medskip

In order to prove this result, we decompose the construction of the trajectory
into several steps. First, we will focus on partial configurations restricted to
the star where one agent is active and one agent is inactive. We will see
that there exists a trajectory between any such partial configurations. In
particular, it is possible to exchange the statuses of two agents. Next, we will prove
  that there is a trajectory from $\omega$ to a configuration where the active
  and inactive statuses are on the branches of the complex star.  Then, we
  will fix the status of every agent by decreasing distance to the center of the
  star. The key point is that while we fix a new status, the previous statuses
  are preserved.


\paragraph{From the complex star to the complex star.}

Let $\mathcal{I}$ be the set of configurations such that at least one node in $\{s_1,s_2,s_3\}$ has value $0$ and at least one node in $\{s_1,s_2,s_3\}$ has value $1$. Formally, we have
\[
\mathcal{I}=(\{s_1\},\{s_2\})^+ \cup (\{s_2\},\{s_1\})^+ \cup (\{s_1\},\{s_3\})^+ \cup (\{s_3\},\{s_1\})^+ \cup (\{s_2\},\{s_3\})^+ \cup (\{s_3\},\{s_2\})^+.
\]

\begin{lemma}\label{startostar}
For every pair of partial configurations $(X',Y')$ and $(X'',Y'')$ in $\mathcal{I}$ there exists a trajectory of even length between $(X',Y')$ and $(X'',Y'')$.
\end{lemma}

\begin{proof}
Consider a restricted configuration on $\{s_1,s_2,s_3\}$. We assume without loss of generality that $s_1$ is active whereas $s_2$ is inactive. We want to prove that we can reach any restricted configuration on $\{s_1,s_2,s_3\}$.

First, let us prove that we can generate both the configuration where $s_3$ is
inactive or active without using any node outside the complex star
$\{s_*,s_1,s_2,s_3,s'_1,s'_2,s'_3\}$.  Let us show that we can reach the case where $s_3$ is infected. We consider the following sequence of restricted configurations:
\begin{itemize}
\item $s_1$ is active and $s_2$ is inactive,
\item $s_*$ is active and $s'_2$ is inactive,
\item $s_1,s_3$ are active and $s_2$ is inactive.
\end{itemize}
We constructed a trajectory between our original configuration and the one we wanted. Notice that in fact, we only use as auxiliary state $s_*$ and $s'_2$, hence there are no constraints on the statuses of $s'_1$ and $s'_3$ (if they are different from $s'_2$). By exchanging the role of $s_1$ and $s_2$, we obtain a procedure such that $s_3$ is inactive.

Once $s_3$ is fixed, it remains to fix by the same argument the statuses of
  $s_1,s_2$: if $s_3$ is active, by considering $s_2$ as the inactive, fix
  the status of $s_1$ as it should be and then finally fix the status of $s_2$;
  if $s_3$ is inactive, exchange $s_1$ with $s_2$.
\end{proof}

\paragraph{From $\omega$ to the complex star.}

We show that there exists a trajectory of odd length from $\omega$ to $\mathcal{I}$.
\begin{lemma}\label{centrage}
  Let $\omega \in \Omega$. Assume that there exist $x,x'\in \scX_o$ such that
  $\omega(x)=1$ and $\omega(x')=0$. Then there exists $n\in 2\NN$ such that
\[
K^n(\omega,\cI)>0.
\]
\end{lemma}

\begin{proof}
By assumption, there exists $x,y\in \mathcal{X}_o$ such that $x$ is active and
$y$ is inactive in configuration $\omega$. We will distinguish by which one is
closer to the center of the star $s_*$. Notice that since $x$ and $y$ are both
in $\mathcal{X}_o$ (we assumed the network is bipartite), their distance to $s_*$ differ by an
even number.  We choose a shortest path $p$ from $x$ to $s_*$.
\begin{itemize}
\item If $d(x,s_*)<d(y,s_*)$ then $p$ does not go through $y$. Consider the shortest path from $y$ to $s_*$. By definition of a complex star, there exists $i\in \{1,2,3\}$ such that $s_i$ is not on the shortest path between $y$ and $s_*$. Let us assume without loss of generality that it is $s_1$. We propagate the active status to $s_1$ through $p$ while storing $y$ on any of its neighbors. By the distance assumption, none of the neighbors of $y$ can be on $p$ or it will contradict the assumption on distances. We distinguish two cases:
\begin{itemize}
\item the shortest path from $y$ to $s_*$ goes through $s_2$ (resp. $s_3$): one propagates the inactive status to $s_2$ (resp. $s_3$) through the shortest path while storing the status of $s_1$ in $s_*$.
\item the shortest path from $y$ to $s_*$ does not go through $\{s_1,s_2,s_3\}$. Therefore, there exists $s_4 \in \Gamma(s_*)$ such that $s_4 \notin \{s_1,s_2,s_3\}$. One can propagate the inactive status to $s_4$ through the shortest path while storing the status of $s_1$ in $s_*$. To conclude, by applying Lemma \ref{startostar} to the star $\{s_*,s_1,s'_1,s_2,s'_2,s_4,s'_4\}$, one can place the inactive status in $s_2$ while preserving the active status at $s_1$.
\end{itemize}  
\item If $d(x,s_*)>d(y,s_*)$ then the proof is similar as the first case by inverting the role of active and inactive agents.
\item If $d(x,s_*)=d(y,s_*)$ then we need to distinguish three cases depending on the neighbors of $y$ and $x$:
\begin{itemize} 
\item If $y$ has only one neighbor, not on $p$: We follow the proof of the first case by storing the status of $y$ on this neighbor. 
\item If $y$ has more than one neighbor: Then, one of them has to be outside of $p$. Indeed, if two of them are on the path $p$, then it means than one of these neighbors is at distance less than $d(x,s_*)-3$ of $s_*$ which would put $y$ at a distance less than $d(x,s_*)-2$, a contradiction. We follow the proof of the first case by storing the status of $y$ on this neighbor. 
\item  If $y$ has only one neighbor which is on $p$: Then by assumption on the
  distances, there exists $z\in \scX$ in the neighborhods of $x$ and $y$. The
  case $\Gamma(x)=\Gamma(y)=\{z\}$ is impossible as the statuses of $x,y$
    cannot be stored in the same node $z$. Hence, $|\Gamma(x)|>1$ and we can treat this case by exchanging the role of $x$ and $y$: first propagate the status of $y$ while storing $x$ on its other neighbor.
\end{itemize}
\end{itemize}
\end{proof}

\paragraph{From the star to the cylinder $(X,Y)^+$.}

By assumption, the complex star is disjoined from $X\cup \overline{X} \cup Y \cup \overline{Y}$. This condition ensures that the storing algorithm for $(X,Y)^+$ from Lemma \ref{global_saving} can be applied without using the star. Hence, we are free to manipulate the status of nodes in the star. 


\begin{lemma}\label{placement_ordonne}
Let $\omega \in \mathcal{I}$. Then for every $z\in \mathcal{X}_o$, one can place an active agent at $z$ in $n\in 2\NN$ stages. Moreover, this can be done without changing the values for nodes at a distance greater than $z$ from $s_*$.
\end{lemma}

\begin{proof}
Let us consider $\omega_*$ a configuration in $\cI$ and $z\in \mathcal{X}_o$ not
on the complex star and with no neighbors on the complex star. Let us say without loss of generality that we want to place an active agent at $z$ and that $s_1$ is active whereas $s_2$ is inactive.

\smallskip

By assumption, there exists a global storing mechanism for any configuration denoted by $\theta$. Moreover, by the condition on the neighborhoods, the status of any agent $z'$ is stored on one of its neighbors which is not in $\{s_*,s_1,s_2,s_3,s'_1,s'_2,s'_3\}$.

\smallskip

Consider a shortest path between $s_*$ and $z$. There are two cases:
\begin{itemize}
\item there exists a shortest path not going through $\{s_1,s_2,s_3,s'_1,s'_2,s'_3\}$, then the procedure is as follows: propagate the active status from $s_1$ to $s_*$ and then to $z$ while storing $s_2$ on $s'_2$ and $s_1$ at $s_*$. On every agent at a distance further away from $z$, store the status by using the global storing mechanism.
\item every shortest path goes through $\{s_1,s_2,s_3,s'_1,s'_2,s'_3\}$. Assume that it goes through $s'_i$, then there also exists a shortest path such that it goes through $s_i$, since by construction $s'_i$ is at distance $2$ from $s_*$. We consider this shortest path. We have several cases:
\begin{itemize}
\item if it goes through $s_1$: then simply propagate $s_1$ to $z$ storing at $s'_1$, and store $s_2$ at $s_*$. Forget about $s_3$.
\item if it goes through $s_2$: first consider the path from the current
  configuration to the one where $s_2$ is active and $s_1$ is inactive by
  Lemma \ref{startostar}. Store $s_1$
   outside the path that we need. Then, we are back to the previous case.
\item if it goes through $s_3$: consider the path from the current configuration to the one where $s_3$ is active and $s_1$ is inactive by Lemma \ref{startostar} and finish like in the previous case.
\end{itemize}
\end{itemize}

The storing mechanism ensures that anything that happens outside of the path and the stars can be preserved.
\end{proof}

\begin{lemma}\label{placement_ordonne1}
Let $\omega \in \cI$. Then for every $z\in \mathcal{X}_o$, one can place an inactive agent at $z$ in $n\in 2\NN$ stages. Moreover, that can be done without changing the values for nodes at a distance greater than $z$ from $x_*$.
\end{lemma}

By iterating the previous lemma by decreasing distance, we obtain Proposition \ref{cylinder_odd} conditionally on the fact that the complex star is centered at an even node. Nevertheless, the previous proof still holds when the star is centered in $\mathcal{X}_o$. Indeed, it is sufficient to notice that now the propagation from $\omega$ to $\cI$ is done in an odd number of stages and the propagation from the star to $\mathcal{X}_o$ (resp. $\mathcal{X}_e$) is done in an odd (resp. even) number of stages, yielding a total of even (resp. odd) number of stages. Hence, Proposition \ref{cylinder_odd} is proven.

\section{Proof of Theorem~\ref{dec}}\label{appD}
We will make use of the following observation:
\begin{observation}\label{pro:nobipartite}
Let $(\scX,\sim)$ be a nonbipartite network. For every subnetwork
$(\scX,\approx)$ with a bipartition $(\scX_o,\scX_e)$, at least one of the two following statements holds:
\begin{itemize}
\item[-] there exists $x\in \scX_o$ such that $\Gamma^\sim(x) \cap \scX_o \neq \varnothing$;
\item[-] there exists $x\in \scX_e$ such that $\Gamma^\sim(x) \cap \scX_e \neq \varnothing$.
\end{itemize}
\end{observation}

The result concerning $(\infty,0)$ and $(0,\infty)$ has already
  been shown in the general case (Proposition~\ref{prop:01}). We denote by 
    $(\scX_o,\scX_e)$ the bipartition of the subgraph $(\scX,\approx)$.

\underline{$(\infty,\infty)$ is absorbing:}\\

We have to prove that if $\omega \in (\infty,\infty)$ and $\omega'$ is randomly
chosen with probability $K(\omega,.)$, then with probability $1$, $\omega'\in
(\infty,\infty)$. Let us compute the probability that $\omega'$ has a finite or
null number of $1$. Denote by
\[
\Theta=\{x \in \mathcal{X}, \ \omega(x)=1\}
\]
then $\Gamma(\Theta)=\cup_{x\in \Theta} \Gamma(x)$ is infinite. But by
assumption, any $x'\in \Gamma(\Theta)$ has at least one neighbor with value $1$,
it follows that the probability that $\omega'(x')=0$ is strictly smaller than
$\rho<1$. Since the probabilities are independent, it follows that the total
probability of infinitely many to be equal to $0$ is $0$. The same reasoning
holds for $0$.\\

\underline{$(\infty,\F)$ and $(\F,\infty)$ are transient:}\\

Let us consider $(\scX,\approx)$ and the partition defined in Section
\ref{sec:decomp_absorbing_transient}. Then, if $\omega\in(\infty,\F)$ then we
know that there are only a finite number of active agents and an
infinite number of inactive agents. We have two possibilities:
\begin{itemize}
\item $\omega \in (\infty,0,\infty,\F) \cup (\infty,\F,\infty,0)$,
\item $\omega \in (\infty,\F,\infty,\F)$.
\end{itemize}
We showed in Section \ref{sec:decomp_absorbing_transient} that both sets are transient (for the subgraph). It follows that $(\infty,\F)$ is also transient for the subgraph and therefore for the original graph.\\

\underline{$(\infty,\infty)$, $(\infty,\F)$ and $(\F,\infty)$ are
  $\phi$-irreducible:}\\ Let $\omega$ be a configuration with one active agent
denoted by $x$ and one inactive agent denoted by $y$. We consider a subgraph
$(\scX,\approx)$ with bipartition $(\scX_o,\scX_e)$. Moreover, by
Proposition \ref{pro:nobipartite}, there exists $c \in \scX$ with a neighbor in
the same block of the bipartition. We will assume in the following that
$c\in \scX_o$ and there exists $d\in \Gamma(c) \cap \scX_o$. Finally, let
$(s_*,s_1,s_2,s_3,s'_1,s'_2,s'_3)$ be a complex star in the graph
$(\scX,\approx)$ centered in $\scX_e$ disjoint from $\{c,d\}$.\\

\noindent Our aim is to construct a trajectory to a partial configuration
$\omega'$ such that there exist both an active agent and an inactive
  agent on some even positions, and also on some odd positions. We distinguish
several cases:
\begin{itemize}
\item If $x\in \scX_o$ and $y\in \scX_o$, then one can apply Lemma
  \ref{centrage} to construct a trajectory between $\omega$ and $\mathcal{I}$,
  i.e., with a $1$ on an odd position (say, $s_1$) and a $0$ on an odd
  position (say, $s_2$). The next step is to propagate the active
  status from $\mathcal{I}$ to $c$ through the subgraph (even number of
  stages), from $c$ to $d$ through their direct link (one stage), and then back
  to $s'_1$ through the subgraph (odd
  number of stages), while storing $s_1$ at $s'_1$ and $s_2$ at $s_*$. The
  total number of stages being even, we obtain at even stages the value 1 at
  $s_1\in\scX_o$ and at $s'_1\in\scX_e$. As these two nodes are neighbors,
    the value 1 can be kept at odd stages as well. We proceed differently for the
    inactive agent $s_2$, reaching $c$ and then $d$ as before in an odd number
    of stages. As in odd stages, 0 is stored at $s_*$, we
    have at odd stages a 0 at $s_*\in\scX_e$ and at $d\in\scX_o$, as desired.
\item If $x\in \scX_e$ and $y\in \scX_e$, then the proof is similar.
\item If $x\in \scX_o$ and $y\in \scX_e$, propagate $x$ to $s_1$, so that
  $s_1$ has value 1 on even stages. Then propagate $y$ to $c$ in an odd number of
  stages through the subgraph, then propagate to $d$ by the direct link between
  $c$ and $d$. Doing so, $d$ has value 0 at even stages. 
We are back to the previous case with $y'=d\in \scX_o$ and $x'=s_1\in \scX_o$.
\end{itemize}

Let $X$ and $Y$ be two finite sets. We want to prove that $(X,Y)^+$ is reachable from $\omega'$. $X$ can be decomposed into $X_o=X\cap \scX_o$ and $X_e=X\cap \scX_e$ and $Y$ can be decomposed into $Y_o=Y\cap \scX_o$ and $Y_e=Y\cap \scX_e$. We can now restrict ourselves to the subgraph $(\scX,\approx)$ and apply the results of Section \ref{sec:decomp_ergodicity}.

 %
%
%

%
%

%

\end{document}